\documentclass{article}
\usepackage[margin=1in]{geometry}
\usepackage{amsmath,amsfonts,amssymb,amsthm,hyperref,braket,todonotes,xcolor,zref-clever,graphicx,physics,mathtools}
\usepackage[
    n,
    operators,
    advantage,
    sets,
    adversary,
    landau,
    probability,
    notions,
    logic,
    ff,
    mm,
    primitives,
    events,
    complexity,
    asymptotics,
    keys]{cryptocode}
\createprocedureblock{procb}{center ,boxed}{}{}{}
\newcommand{\pw}{\mathsf{pw}}
\newcommand{\pwg}{\mathsf{pwguess}}
\newcommand{\NewS}{\mathsf{NewSession}}
\newcommand{\role}{\mathsf{role}}
\newcommand{\sid}{\mathsf{sid}}
\newcommand{\TestPwd}{\mathsf{TestPwd}}
\newcommand{\NewKey}{\mathsf{NewKey}}
\newcommand{\Hyb}{\mathsf{Hyb}}
\newcommand{\Xr}{\mathsf{X}}
\newcommand{\Yr}{\mathsf{Y}}
\newcommand{\Zr}{\mathsf{Z}}
\newcommand{\Ar}{\mathsf{A}}
\newcommand{\Br}{\mathsf{B}}
\newcommand{\Cr}{\mathsf{C}}
\newcommand{\Wr}{\mathsf{W}}
\newcommand{\Kr}{\mathsf{K}}
\newcommand{\Rr}{\mathsf{R}}
\newcommand{\Dr}{\mathsf{D}}
\newcommand{\Or}{\mathsf{O}}
\newcommand{\Fr}{\mathsf{F}}
\newcommand{\Sreg}{\mathsf{S}}
\newcommand{\Treg}{\mathsf{T}}
\newcommand{\PK}{\mathcal{PK}}
\newcommand{\SK}{\mathcal{SK}}

\newcommand{\B}{\mathcal{B}}
\newcommand{\FPake}{\mathcal{F}_{\text{PAKE-ea}}}
\DeclareMathOperator*{\E}{\mathbb{E}}
\newcommand{\cA}{\mathcal{A}}
\newcommand{\cD}{\mathcal{D}}
\newcommand{\cE}{\mathcal{E}}
\newcommand{\cK}{\mathcal{K}}
\newcommand{\cX}{\mathcal{X}}
\newcommand{\cY}{\mathcal{Y}}
\newcommand{\cT}{\mathcal{T}}
\newcommand{\inp}{\mathsf{in}}
\newcommand{\good}{\mathsf{good}}
\newcommand{\Amp}{\mathsf{Amp}}

\newcommand{\pure}{\mathsf{pure}}
\newcommand{\clock}{\mathsf{clock}}
\newcommand{\forward}{\mathsf{for}}
\newcommand{\rev}{\mathsf{rev}}

\newcommand{\acc}{\mathsf{acc}}

\newcommand{\Hkey}{H_{\mathsf{key}}}
\newcommand{\Htag}{H_{\mathsf{tag}}}

\newcommand{\fresh}{\mathsf{fresh}}
\newcommand{\interrupted}{\mathsf{interrupted}}
\newcommand{\compromised}{\mathsf{compromised}}
\newcommand{\completed}{\mathsf{completed}}

\newcommand{\PAKE}{\mathsf{PAKE}}
\newcommand{\chall}{\mathsf{chall}}
\newcommand{\KEM}{\mathsf{KEM}}
\newcommand{\Enc}{\mathsf{KEM.Encaps}}
\newcommand{\Dec}{\mathsf{KEM.Decaps}}
\newcommand{\KG}{\mathsf{KEM.KeyGen}}
\newcommand{\Env}{\mathcal{Z}}
\newcommand{\Sim}{\mathcal{S}}
\newcommand{\R}{\mathcal{R}}
\newcommand{\pkey}{\mathsf{key}}
\newcommand{\ICEnc}{\mathcal{E}}
\newcommand{\ICDec}{\mathcal{D}}
\newcommand{\IC}{\mathsf{IC}}

\newcommand{\TF}{\mathsf{2F}}
\newcommand{\msg}{\mathsf{msg}}
\newcommand{\flag}{\mathsf{flag}}
\newcommand{\Real}{\mathsf{Real}}
\newcommand{\Ideal}{\mathsf{Ideal}}
\newcommand{\chck}{\mathsf{check}}

\newcommand{\Sen}{\mathsf{Sen}}
\newcommand{\Rec}{\mathsf{Rec}}

\newcommand{\cF}{\mathcal{F}}
\newcommand{\OT}{\mathsf{OT}}
\newcommand{\Com}{\mathsf{Com}}
\newcommand{\Open}{\mathsf{Open}}

\newtheorem{theorem}{Theorem}[section]
\newtheorem{corollary}{Corollary}[theorem]
\newtheorem{lemma}[theorem]{Lemma}
\newtheorem{claim}[theorem]{Claim}
\newtheorem{definition}[theorem]{Definition}

\AddToHook{env/lemma/begin}{\zcsetup{countertype={theorem=lemma}}}
\AddToHook{env/claim/begin}{\zcsetup{countertype={theorem=claim}}}
\AddToHook{env/definition/begin}{\zcsetup{countertype={theorem=definition}}}
\zcRefTypeSetup{claim}{
  name-sg=claim, name-pl=claims,
  Name-sg=Claim, Name-pl=Claims
}
\title{Natural Barriers to Quantum Extraction: \\ \Large{On the Post-Quantum (In)security of (O)EKE and Masny-Rindal OT}}
\author{James Bartusek\thanks{Columbia University. \href{mailto:bartusek.james@gmail.com}{bartusek.james@gmail.com}. } \and Jake Januzelli\thanks{Columbia University. \href{mailto:jj3544@columbia.edu}{jj3544@columbia.edu}}}
\date{}

\newif\ifnotes
\notestrue

\begin{document}

\pagenumbering{roman}
\maketitle

\begin{abstract} 
Encrypted key exchange (EKE), introduced by Bellovin and Merritt (IEEE S\&P 1992), and Masny-Rindal OT, introduced by Masny and Rindal (ACM CCS 2019), are highly-efficient methods for compiling essentially any KEM into advanced cryptographic protocols, namely password-authenticated key exchange (PAKE) and oblivious transfer (OT), by relying only on idealized symmetric-key primitives. They have become leading candidates for practically-implementable PAKE and OT due to (1) their simplicity, (2) their plug-and-play nature, allowing for flexibility in the choice of KEM, and (3) existing proofs of UC-security (in the classical adversarial model).  

Due to point (2) above, these compilers yield attractive candidates for efficient \emph{post-quantum} PAKE and OT, especially given the recent post-quantum KEM standardization efforts. This motivates the question of whether the (UC-)security of these compilers translates to the quantum adversarial model. In this work, we show that it does not. In particular, we prove that a general family of (O)EKE protocols, as well as Masny-Rindal OT, are \emph{not} UC-secure against quantum polynomial-time adversaries, even when instantiated with a post-quantum KEM. To establish UC-insecurity, we devise an adversarial strategy that provably thwarts any attempt by the simulator to extract its input (the password in the case of PAKE, and the receiver's choice bit in the case of OT).

To complement these negative results, we establish that both compilers yield certain notions of \emph{game-based} security. In the PAKE setting, we consider OEKE instantiated with the ``2-Feistel'' cipher, and prove its stand-alone game-based security in the quantum random oracle model.  We view these results as a proof of concept that security of (O)EKE and Masny-Rindal OT can yet be redeemed in the quantum setting, though we caution that the situation will be more subtle than in the classical setting due to the breakdown of simulation-based security. Along the way, we establish a novel ``advantage-tight'' one-way to hiding lemma that may be of independent interest.

\end{abstract}

\newpage
\tableofcontents
\thispagestyle{empty}
\newpage
\pagenumbering{arabic}

\section{Introduction}
The threat of quantum computers has precipitated widespread adoption of post-quantum cryptography in recent years. Moving  from cryptographic assumptions vulnerable to Shor's algorithm (like factoring or discrete logarithm) to post-quantum assumptions is the first step in this transition. However, this does not always suffice to establish security against quantum attackers: Classical security proofs often make additional assumptions about the adversary that can fail in the quantum setting (e.g. that one can rewind the adversary \cite{FOCS:AmbRosUnr14}).

Practically every efficient post-quantum primitive relies on the Random Oracle Model (ROM) \cite{CCS:BelRog93} to justify security, often through transforms like Fiat-Shamir \cite{AC:PoiSte96} and Fujisaki-Okamoto \cite{C:FujOka99}. In 2011, \cite{AC:BDFLSZ11} introduced the Quantum Random Oracle Model (QROM), which allows quantum adversaries to query a random oracle in superposition (an ability that any real-life quantum attacker will have by implementing a reversible circuit for the hash function in question). Although many important constructions have now been shown secure in the QROM (e.g. \cite{C:Zhandry19,C:DFMS19,C:DonFehMaj20}), these proofs require expertise in quantum information theory and can often be substantially more complex than in the ROM. In light of this, one may ask whether expending the effort to write a QROM proof is actually necessary: Are there any cryptographic protocols in use that are secure in the ROM but insecure in the QROM? The seminal works of Yamakawa and Zhandry \cite{EC:YamZha21,FOCS:YamZha22} construct a variety of theoretical examples including one-way functions and public key encryption, but they are not naturally occurring, and were specifically created with this aim in mind. 

In general, there appears to be a belief in the community that for any ``naturally-occurring'' ROM-based protocol instantiated with post-quantum assumptions, if we currently lack a QROM proof, then this likely indicates a lack of appropriate proof techniques, not an attack. As an illustrative example, a report from the European Telecommunications Standards Institute (one of three official standards organizations for the EU) refers to the \cite{EC:YamZha21,FOCS:YamZha22} counterexamples as ``completely artificial'' and notes that
\begin{center}\textit{``it is often stated \cite{C:DFMS19} that a concrete protocol that is secure in the ROM will remain secure in the Q-ROM (if based on quantum-safe hardness assumptions), although with no formal proof in this sense.''}
\cite[Section~7.1]{ETSI}
\end{center}
\paragraph{This work.} We challenge this claim, and give the first examples of \emph{naturally occurring} cryptographic protocols based on post-quantum assumptions that are secure in the ROM but provably \emph{not} secure in the QROM. Concretely, we show that the oblivious transfer (OT) protocol of Masny and Rindal \cite{masnyrindal} and the (Once) Encrypted Key Exchange, or (O)EKE \cite{SP:BelMer92,CCS:BreChePoi03}, family of protocols for Password Authenticated Key Exchange (PAKE) are insecure in the Universally Composable (UC) framework \cite{FOCS:Canetti01} when quantum access to the random oracle is allowed, despite being UC-secure in the ROM. As a bonus, we show the same insecurity for variants of OEKE in the Quantum Ideal Cipher Model (QICM).

We complement these negative results with positive ones, showing that both protocols do satisfy some notions of game-based security in the QROM under the same assumptions used for the classical ROM proofs. Prior to this work, there were no fully-established results on either the simulation-based or game-based security of either protocol against quantum adversaries.\footnote{Although \cite{hovelmanns2025cake} attempted an analysis of the post-quantum game-based security of OEKE, their handling of the quantum-accessible ideal cipher was only conjectured, not proven.}
\paragraph{Masny-Rindal OT.}
Recall that in 1 out of 2 (random) Oblivious Transfer (OT), a sender outputs two strings $s_0, s_1$ and a receiver outputs $b, s_b$ for a choice bit $b$. Security requires the receiver to learn nothing about $s_{1-b}$ and the sender to learn nothing about $b$. 

In 2019, Masny and Rindal \cite{masnyrindal} introduced a highly efficient compiler that upgrades any Key Encapsulation Mechanism (KEM) protocol satisfying certain natural properties to a two-message random OT protocol, in the ROM.\footnote{In fact, their techniques also yield a \emph{simultaneous-message} OT protocol in some settings, but in this work we focus on the two-message version.} The compiler is described in \zcref[S]{OTdiag} and can be instantiated with ML-KEM, making it the premier candidate for deployable post-quantum OT. In particular, their results establish that the protocol is UC-secure\footnote{The functionality that they prove security with respect to is called \emph{endemic} OT, which is the weakest known notion of random OT that remains meaningful. Here, meaningful refers to the fact that any random OT satisfying endemic security implies a simulation-secure chosen-message OT via the natural compiler in which the sender transmits $m_0 \oplus s_0, m_1 \oplus s_1$ along with its second message, where $m_0$ and $m_1$ are its chosen strings.} in the ROM when instantiated with ML-KEM. In this work, we show that this does not hold true in the \emph{quantum} random oracle model.

\begin{theorem}[Informal]
    Masny-Rindal OT instantiated with \emph{any} $\KEM$ is not post-quantum simulation-secure.
\end{theorem}

Our attack exploits an inability for the simulator to \emph{extract} a malicious receiver's choice bit. However, it does not establish that the malicious receiver can actually obtain both of the sender's messages. Hence, it does not rule out the possibility that Masny-Rindal actually does satisfy some game-based notion of security. We investigate this possibility, and show that it does.

\begin{theorem}[Informal]
    Let $\KEM$ be any KEM that satisfies (computational) public-key uniformity and key unpredictability. Then Masny-Rindal OT instantiated with $\KEM$ satisfies the following claim in the QROM: The probability that any QPT malicious receiver can guess $s_{b'}$ for random $b'$ is at most $1/2 + \mathsf{negl}$.
\end{theorem}

\begin{figure}[ht]
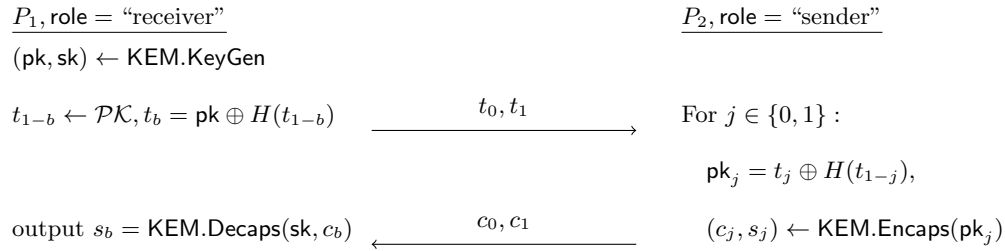

\pseudocodeblock{
\underline{P_1, \role = \text{``receiver''}} \< \< \underline{P_2, \role = \text{``sender''}} \\ 
(\pk,\sk) \gets \KG \< \< \\
t_{1-b} \gets \PK, t_b = \pk \oplus H(t_{1-b}) \< \sendmessageright*{t_0, t_1} \< \text{For } j \in \{0,1\}: \\
\< \< \quad \pk_j = t_j \oplus H(t_{1-j}),  \\
\text{output } s_b = \Dec(\sk,c_b) \< \sendmessageleft*{c_0, c_1} \< \quad (c_j,s_j) \gets \Enc(\pk_j)
}
\caption{Masny-Rindal OT with random oracle $H: \PK \to \PK$, a KEM $\KEM$, and receiver's choice bit $b \in \{0,1\}$.}
\label{OTdiag}
\end{figure}
\paragraph{PAKE.} In the PAKE setting parties hold (possibly low entropy) passwords $\pw, \pw'$, and want to run an authenticated key exchange based on them: if $\pw = \pw'$ they agree on a shared secret key, and if not they derive independent keys (or output $\bot$). PAKE protocols offer a form of authentication without public key infrastructure, and have been deployed in password managers, end-to-end encrypted backup recovery, key escrow, and WiFi login.\footnote{See \url{https://github.com/fancy-cryptography/fancy-cryptography} for more information.} Although there are both game-based \cite{EC:BelPoiRog00} and UC \cite{EC:CHKLM05} definitions of PAKE security, the UC definition has become the gold standard that deployed protocols target, due to modeling password correlation between parties and security under arbitrary composition: all Round 2 candidates of the 2020 CFRG PAKE competition have proofs of UC-security \cite{CFRG_PAKE_Selection}.

We study the dominant paradigm of post-quantum PAKE: Once Encrypted Key Exchange (OEKE) (depicted in \zcref[S]{OEKEdiag}).  Although originally instantiated with an ideal cipher, recent works  \cite{CCS:McQRosRoy20,EC:JanRoyXu25,arriaga2026tempo} opt to replace the ideal cipher with a two-round Feistel network using a random oracle. This 2-Feistel variant of OEKE is widely favored due to its optimal round complexity, efficiency, and reliance only on the ROM and a post-quantum KEM. In fact, a variant of this protocol is part of an Internet-Draft submitted to the Internet Engineering Task Force for possible standardization \cite{vos_hybrid_draft}. In the corresponding paper \cite{vos2025hybrid}, the authors prove their protocol is UC-secure in the (classical) ROM, conjecturing that their protocol is also secure in the QROM:

\begin{center}\textit{``A crucial next step for theoretical security is to prove security in the quantum random oracle model. We believe that this is possible using compressed quantum random oracles, because the simulated random oracles only rely on historical queries in simple ways.''}
\cite[Section~6]{vos2025hybrid}
\end{center}

\begin{figure}[ht]
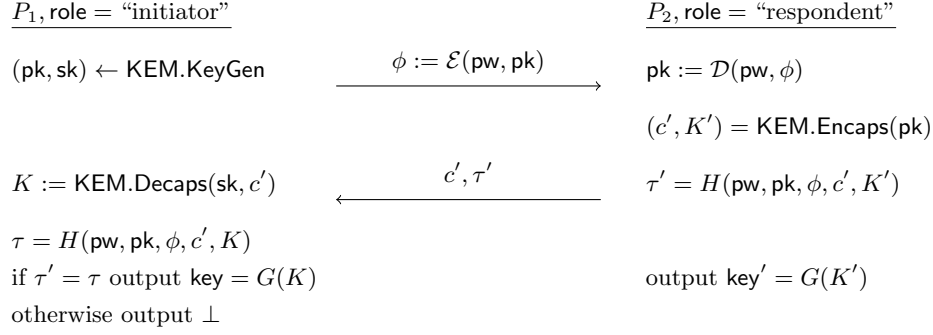


\pseudocodeblock{
\underline{P_1, \role = \text{``initiator''}} \< \< \underline{P_2, \role = \text{``respondent''}} \\ 
(\pk,\sk) \gets \KG \< \sendmessageright*{\phi := \ICEnc(\pw,\pk)} \< \pk := \ICDec(\pw,\phi) \\
\< \< (c',K') = \Enc(\pk) \\
K := \Dec(\sk,c') \< \sendmessageleft*{c',\tau'} \< \tau' = H(\pw, \pk, \phi, c',K') \\
\tau = H(\pw, \pk, \phi, c',K) \< \< \\
\text{if } \tau' = \tau \text{ output }\pkey = G(K) \< \< \text{output }\pkey' = G(K') \\
\text{otherwise output } \bot  \< \<
}
\caption{A typical instantiation of OEKE with shared password $\pw$, random oracles $H,G$, a KEM $\KEM$, ideal cipher encryption $\ICEnc$ and decryption $\ICDec$.}
\label{OEKEdiag}
\end{figure}

We refute this conjecture, and show that OEKE is UC-insecure in the QROM.

\begin{theorem}[Informal]
    OEKE instantiated with \emph{any} $\KEM$ and using either ideal cipher encryption or 2-Feistel encryption is not post-quantum UC-secure.
\end{theorem}

In \zcref[S]{subsec:oekeattack} we prove the attack works for the ideal cipher instantiation, as this suffices to illustrate the main idea. In \zcref[S]{subsec:attack-general} we give some details as to how it generalizes to a wide variety of OEKE variants in the QROM: this includes NoIC \cite{arriaga2025noic}, Tempo \cite{arriaga2026tempo}, OEKE/EKE \cite{EC:JanRoyXu25}, OQUAKE \cite{vos2025hybrid} and (O)CAKE \cite{ACNS:BCPRR23,hovelmanns2025cake}.

Similar to the Masny-Rindal OT case, our attack relies on an inability for the simulator to extract (in this case, the password) from a malicious initiator. Therefore, this also leaves open the possibility of game-based security. We prove some positive results when the encryption used is 2-Feistel.\footnote{We work with 2-Feistel since (1) it is more efficient and easier to instantiate than an ideal cipher, and (2) the QROM is easier to work with than the QICM.}

\begin{theorem}[Informal]
    Let $\KEM$ be any KEM that satisfies (computational) public-key uniformity and key unpredictability. Then OEKE instantiated with $\KEM$ and using 2-Feistel encryption satisfies stand-alone game-based security in the QROM against both malicious initiators and malicious respondents.\footnote{For those familiar with BPR security \cite{EC:BelPoiRog00} see \zcref[S]{sec:future}.}
\end{theorem}

\paragraph{Advantage-tight one-way to hiding.} Along the way to our positive results on game-based OEKE security, we establish a new one-way to hiding lemma that to the best of our knowledge has not appeared before, and is of independent interest. One-way to hiding lemmas \cite{EC:Unruh14,C:AmbHamUnr19} are essentially search-to-decision reductions with respect to quantum-accessible oracles: Given a distribution over two classical functions $G$ and $H$ that differ on a set $S$, they relate the maximum probability $p$ that an adversary can find an $x \in S$ given quantum access to $G,H$ to the maximum distinguishing advantage $f$ that an adversary has given access to $G$ or $H$. A typical statement (e.g. \cite[Lemma~3.3]{EC:KSSSS20}) upper bounds $f$ as a function of $p$ and the number of queries $q$ made by the adversary (or its query depth $d$).

Here, we show that $f$ is upper-bounded by $p$ \emph{itself} (and independent of $q$ or $d$), up to an additive negligible term. This statement comes with two caveats: (1) The set $S$ must be \emph{efficiently checkable} (though this is also required to obtain $f \leq p$ in the classical setting), and (2) the reduction (search adversary) may run in time much longer than (but still polynomial in) the decision adversary, in particular, the multiplicative blowup can depend on the number of queries $q$.

Nevertheless, we establish a completely ``advantage-tight'' one-way to hiding theorem as stated below, which turns out to be crucial for establishing the tight post-quantum game-based security of OEKE.

\begin{theorem}[Informal]
    Let $D$ be any sampler that outputs functions $H,G$ that differ on a set $S$, along with an arbitrary quantum state $\rho$. Suppose that for all QPT $B$,  \[\Pr_{H,G,S,\rho}[x \in S : x \gets B^{H,G,M_S}(\rho)] \leq p(\lambda) + \mathsf{negl}(\lambda),\] where $M_S$ is the membership-checking predicate for the set $S$. Then for all QPT $A$,\[\Big|\Pr_{H,G,S,\rho \gets D}[A^H(\rho) = 1] -  \Pr_{H,G,S,\rho \gets D}[A^G(\rho) = 1]\Big| \leq p(\lambda) + \mathsf{negl}(\lambda).\] 

\end{theorem}

\section{Technical overview}

\subsection{An illustrative commitment scheme}

Our main techniques can be appreciated by studying a simple bit commitment scheme, divorced from the particulars of OT and PAKE. In particular, we abstract out a commitment scheme underlying the Masny-Rindal OT construction \cite{masnyrindal} and show that it satisfies the following combination of properties, in addition to being computationally hiding.

\begin{enumerate}
    \item Extractable against classical committers.
    \item Computationally collapse-binding (the post-quantum analogue of standard binding). 
    \item Provably not extractable against quantum committers.
\end{enumerate}

In this overview, we'll present the commitment scheme, argue why it is \emph{not} quantumly-extractable, and then argue why it \emph{is} collapse-binding. We then touch on how these proofs of non-extractability and collapse-binding are at the heart of our results showing that Masny-Rindal OT and OEKE are not simulation-secure against quantum adversaries but do satisfy some notions of game-based security.

\paragraph{The scheme.} Let $G : \{0,1\}^\lambda \to \{0,1\}^{n}$ be an injective PRG, and let $H : \{0,1\}^{n} \to \{0,1\}^{n}$ be a random oracle. To commit to a bit $b$, sample a seed $s_b \gets \{0,1\}^\lambda$, a string $t_{1-b} \gets \{0,1\}^{n}$, and define $t_b \coloneqq G(s_b) \oplus H(t_{1-b})$. The commitment consists of the two strings $(t_0,t_1)$. To open, send $s_b$, allowing the receiver to verify that $t_b = G(s_b) \oplus H(t_{1-b})$.

It is not hard to see that $(t_0,t_1)$ computationally hides the committer's choice of $b$, since, after switching $G(s_b)$ to uniform, both $t_0$ and $t_1$  are distributed uniformly at random. Moreover, \cite{masnyrindal} shows as a component of their OT security proof that one can extract the bit $b$ from any classical committer by observing the \emph{order} in which they make their random oracle queries. In particular, the committer can ``program'' $t_b$ by querying $t_{1-b}$ first, but then $H(t_b) \oplus t_{1-b}$, which is supposed to equal $G(s_{1-b})$, becomes uncontrollable. Therefore, if $t_{1-b}$ was queried first, the only bit that the committer can hope to open is $b$.

\subsection{Failure to quantumly extract}

The first hint that this extraction strategy may fail in the quantum setting is an inability to define ``query order'' in any meaningful way. Indeed, the quantum random oracle allows for superposition queries, meaning that an adversarial committer could query $H(t_0)$ and $H(t_1)$ at the same time\footnote{The same observation was made in \cite[Section~4~3.1.2]{verschoor2022quantum}.}. One natural attempt to thwart the extractor is to simply run the honest committer's strategy coherently over a uniform choice of the bit $b$, and delay the measurement of $b$ until the opening stage.

\paragraph{Definition of extractability.} Before further investigating this strategy, we recall the formal notion of an extractable commitment. First, any adversarial committer $\cA$ is split into two parts: (1) the commit stage adversary $\cA_\Com$ is initialized with some state $\rho$ on register $\Cr$ and outputs $(t_0,t_1)$ and an updated state $\rho^\Real_\Cr$, and (2) the opening stage adversary $\cA_\Open$ takes $\rho^\Real_\Cr$ and outputs a bit $b$ and a candidate opening $s_b$.  

For any such adversary, there must exist an extractor that takes as input $\rho_\Cr$ and produces a bit $b^*$ and a state $\rho^\Ideal_\Cr$ satisfying the following two properties:
\begin{enumerate}
    \item The probability that $\cA_\Open(\rho^\Ideal)$ produces a successful opening to $1-b^*$ is negligible.
    \item $\rho^\Ideal_\Cr$ is computationally indistinguishable from the state $\rho^\Real_\Cr$ output by the real committer. In fact, this must hold even in the presence of a register $\Dr$ that may be entangled with $\rho_\Cr$. This indistinguishability in the presence of an auxiliary state (sometimes referred to as the ``environment'') is a standard property in simulation-based security and we will use $\Dr$ in our impossibility result.
\end{enumerate}

We allow the extractor to make arbitrary (potentially non-black-box) use of the code of $\cA_\Com,\cA_\Open$ and respond to the adversary's random oracle queries using any strategy.

\paragraph{The coherent honest committer.} Now, coming back to our candidate commitment, we see that running the honest commitment strategy in equal superposition over the choice of $b$ gives (omitting normalization)
\begin{align*}&\ket{0}_\Br\sum_{s_0 \in \{0,1\}^{\lambda},t_1 \in \{0,1\}^{n}}\ket{s_0}_\Sreg\ket{G(s_0) \oplus H(t_1)}_{\Treg_0}\ket{t_1}_{\Treg_1} \\ + &\ket{1}_\Br\sum_{s_1 \in \{0,1\}^{\lambda},t_0 \in \{0,1\}^{n}}\ket{s_1}_\Sreg\ket{t_0}_{\Treg_0}\ket{G(s_1) \oplus H(t_0)}_{\Treg_1}.\end{align*}
Measuring registers $\Treg_0,\Treg_1$ yields a commitment string $(t_0,t_1)$ and a left-over committer state
\[\sum_{b : t_b \oplus H(t_{1-b}) \in \mathsf{Im}(G)}\ket{b}_\Br\ket*{s_b = G^{-1}(t_b \oplus H(t_{1-b}))}_\Sreg,\]

which is then passed to $\cA_\Open$. It appears we may already have a simple counterexample to extractability: $\cA_\Open$ can simply measure in the standard basis to obtain a valid $(b,s_b)$, each with probability 1/2 (where the probability is taken over its entire execution). So if the extractor makes a guess $b^*$ before this point, it seems that it will be wrong with probability 1/2, violating the property (1) above.

However, recall that the extractor has full access to the committer's state on registers $\Br,\Sreg$ and can apply an arbitrary operation before sending it to $\cA_\Open$, as long as this operation preserves computational indistinguishability. A natural extraction attempt would therefore have the extractor \emph{pre-measure} this state in the standard basis, producing $\ket*{b^*,s_{b^*}}$, and fixing the single choice of $b^*$ that $\cA_\Open$ will be able to open. In fact, this extractor \emph{does} satisfy property (2): In the following sub-section we will argue that no QPT adversary will be able to recover both strings $s_0$ and $s_1$ and thus, by the quantum equivalence of mapping and distinguishing, \[\frac{1}{\sqrt{2}}\left(\ket{0,s_0} + \ket{1,s_1}\right) \approx_c \frac{1}{2}\ketbra{0,s_0} +\frac{1}{2}\ketbra{1,s_1}.\] Hence, we will have to alter the adversarial strategy if we hope to rule out the existence of \emph{any} extractor.
\paragraph{Adding a dummy branch.} One way to see the issue with the above attempt is that the distinguisher does not have a computationally accessible description of the committer's state, and thus cannot detect if the extractor is messing with it. To remedy this, we consider an adversarial strategy that runs the following two strategies with $1/2$ probability: (0) a ``dummy'' strategy that replaces $G(s_0), G(s_1)$ with uniform strings $u_0,u_1$, and (1) the coherent honest committer's strategy given above. The key is that branch (1) will enable successful opening of either $b = 0$ or $b = 1$, while branch (0) will enable the distinguisher to efficiently check that the extractor did not collapse the state. Even though successful opening and efficient checking cannot be performed simultaneously, we will be able to argue that the extractor cannot distinguish between the two branches and thus cannot tailor their strategy to one or the other. Here, we make use of the inaccessible distinguisher register $\Dr$, as follows.

Consider the following mixed state over the uniformly random choice of branch $c \in \{0,1\}$, again omitting normalization: 
\begin{align*}
    &\text{if } c = 0: \quad \sum_{b \in \{0,1\},u_b \in \{0,1\}^{n},t_{1-b} \in \{0,1\}^{n}}\ket{u_b}_{\Dr}\ket{b}_\Br\ket{u_b,t_{1-b}}_\Treg, \\ &\text{if } c = 1: \quad \sum_{b \in \{0,1\},s_b \in \{0,1\}^\lambda,t_{1-b} \in \{0,1\}^{n}}\ket{s_b}_{\Dr}\ket{b}_\Br\ket{G(s_b),t_{1-b}}_\Treg. 
\end{align*}

Here, $\Dr $ is the register held by the distinguisher, and $\Cr = (\Br,\Treg)$ is given to the committer. Completing the honest committer's strategy coherently and measuring the resulting commitment $(t_0,t_1)$ yields the left-over state
\begin{align*}
    &\text{if } c = 0: \quad \frac{1}{\sqrt{2}}\left(\ket{0,t_0 \oplus H(t_1)} + \ket{1,t_1 \oplus H(t_0)}\right), \\ &\text{if } c = 1: \quad \sum_b\beta_b\ket{b,s_b},
\end{align*}
for some weights $\beta_b$ whose values will not be important for the argument.  If $c = 0$, the distinguisher uses its knowledge of $t_0,t_1$ to project its remaining state onto $\frac{1}{\sqrt{2}}(\ket{0,t_0 \oplus H(t_1)} + \ket{1,t_1 \oplus H(t_0)})$. Since this measurement accepts with probability 1 in the real execution, any valid extractor must output a state in which this measurement accepts with probability $1-\mathsf{negl}(\lambda)$. Then we can appeal to the pseudorandomness of $G$ to argue that the extractor, who just sees register $\Cr$, cannot distinguish between the two branches, and hence, the state on branch $c = 1$ can be undetectably replaced with the state on branch $c = 0$. At this point, measuring the $\Br$ register of the state produces a uniformly random bit $b$, \emph{independent} of whatever choice $b^*$ the extractor fixed earlier in the experiment, yielding our final contradiction.

\subsection{Collapse-binding}

Despite the provable inability to extract an adversarial committer's bit, we show that this commitment scheme \emph{does} retain meaningful security against quantum adversaries. In particular, the commitment satisfies collapse-binding, which is the standard notion of computational post-quantum binding.

Below, we'll just argue that it satisfies post-quantum ``classical-style'' binding, which stipulates that it is hard for any QPT adversary to output a commitment $(t_0,t_1)$ along with valid openings $s_0,s_1$ for each bit, meaning that $G(s_0) = t_0 \oplus H(t_1)$ and $G(s_1) = t_1 \oplus H(t_0)$. Standard techniques (e.g. \cite{dall2023necessity}) can be used to show that, since there is only one valid opening per bit,\footnote{One could also generalize the construction to use any $G$ that is one-way and \emph{collapsing}, but we stick with the case of an injective PRG, which trivially satisfies both properties, for the purpose of this overview.} classical-style binding of this commitment suffices to establish collapse-binding, which is equivalent to the following notion of sum-binding: After the commit stage, the challenger samples a $b^* \gets \{0,1\}$, and we require that no QPT adversary can output a valid opening $s_{b^*}$ except with probability $1/2 + \mathsf{negl}(\lambda)$.

\paragraph{Measure-and-reprogram.} We show classical-style binding using an appropriate application of the measure-and-reprogram technique \cite{C:DFMS19,C:DonFehMaj20}, which will allow us to \emph{program} in an external challenge $G(s)$ for an unknown $s \gets \{0,1\}^\lambda$. If the adversary then manages to output $s$ as one of their pair $(s_0,s_1)$, they immediately contradict the pseudorandomness of the PRG $G$.

In more detail, measure-and-reprogram considers the following scenario. An adversary queries a random oracle $H$ and then outputs a set of $k$ oracle inputs $\{x_i\}_{i \in [k]}$ and some auxiliary information $z$. Consider any predicate $V$ over $\{x_i\}_{i}, \{H(x_i)\}_{i}$, and $z$, and suppose that $\Pr[V(\{x_i\}_i,\{H(x_i)\}_i,z) = 1] = \epsilon$. Then, consider a \emph{simulated} experiment where the adversary is run as in the real experiment except that the input register to $k$ randomly-chosen queries are fully measured in the standard basis to yield $\{x_i^*\}_{i}$, and the $H(x_i^*)$ are re-programmed to fresh uniformly random outputs $\{y^*_i\}_{i}$. Essentially (eliding certain details for simplicity of presentation), it holds that in this simulated experiment, $\Pr[V(\{x^*_i\}_i,\{y_i^*\}_i,z) = 1] \approx \frac{\epsilon}{q^{2k}}$, where $q$ is the number of queries made by the adversary.

For arguing the classical-style binding of our commitment, we consider the following predicate: $V$ takes as input the commitment $(t_0,t_1)$ and the openings $(s_0,s_1)$ and outputs 1 iff $G(s_0) = t_0 \oplus H(t_1)$ and $G(s_1) = t_1 \oplus H(t_0)$. Thus, $k = 2$ in this example. After moving to the simulated experiment, we can consider the following re-programming strategy: Upon measuring the first input $t^*_1$, program $H(t^*_1) = y^*$ for a uniformly random $y^*$, and upon measuring the second input $t^*_2$, program $H(t^*_2) = t^*_1 \oplus G(s)$, where $G(s)$ is sampled externally and given to the simulator. The measure-and-reprogram lemma establishes that $V$ still passes with $1/\mathsf{poly}(\lambda)$ probability, which means that $s \in \{s_0,s_1\}$, thus contradicting the security of $G$.

\subsection{Applications}

Next, we argue how the techniques described above can be applied to both Masny-Rindal OT and OEKE to establish their simulation-based insecurity and game-based security.

\paragraph{Masny-Rindal OT: Simulation-based insecurity.} The first message of the  \cite{masnyrindal} protocol can be seen as taking the above commitment scheme and replacing the PRG $G$ with the key generation function $\mathsf{KeyGen}(1^\lambda;r) \to (\pk,\sk)$ of a KEM. That is, the receiver samples random coins $r$, uses them to (deterministically) generate a public key $\pk$, and applies the rest of the commitment algorithm on choice bit $b$ and string $\pk$ to obtain $(t_0,t_1)$. Given $(t_0,t_1)$, the OT sender derives the two possible public keys $\pk_0 = t_0 \oplus H(t_1), \pk_1 = t_1 \oplus H(t_0)$ and then generates and returns ciphertexts $(c_0,c_1)$ encrypting output keys $K_0,K_1$.

The impossibility of extracting a  malicious receiver's choice bit $b^*$ follows essentially the same argument as outlined above for the commitment scheme, with the function $\mathsf{KeyGen}: r \to (\pk,*)$ in place of the PRG $G: s \to G(s)$. Arguing that no simulator can distinguish between the receiver's real and dummy branches still requires assuming that $\pk$ is computationally indistinguishable from random. However, this assumption comes for free: If $\pk$ was distinguishable from uniform, then the OT protocol would be insecure against a malicious \emph{sender}. Indeed, the sender can compute the two candidate public keys $\pk_0,\pk_1$, where $\pk_b$ is generated by $\mathsf{KeyGen}$, and $\pk_{1-b}$ is sampled uniformly at random. Distinguishing $\pk_b$ from random with any non-negligible advantage thus gives the same advantage in guessing the honest receiver's choice bit $b$.

\paragraph{Masny-Rindal OT: Game-based security.} Our argument for the game-based security of the protocol broadly follows the measure-and-reprogram argument introduced above. Against a malicious receiver, we ask that the following game cannot be won with probability better than $1/2 + \mathsf{negl}(\lambda)$: After the honest sender returns $c_0,c_1$, the challenger samples a bit $b^* \gets \{0,1\}$ and the receiver succeeds if they manage to output a correct guess for $K_{b^*}$. Note that this can be seen as sum-binding for a commitment where $K_{b^*}$ is regarded as the valid opening. In a similar manner as above, we can use measure-and-reprogram to reduce to the security of the KEM, though the game is not simply inverting the function $\mathsf{KeyGen}$. In particular, the reduction will ultimately need to take $\pk_{b^*}$ and a ciphertext $c_{b^*}$ from its external challenger, where $c_{b^*}$ must be transmitted to the adversary as a part of the sender's OT message. This actually introduces a slight subtlety in the argument: The reprogramming of the random oracle $H$ must occur \emph{before} this message is sent. Since the adversary may continue to make random oracle queries after the message, we appeal to a version of measure-and-reprogram in which the predicate $V$ may also make oracle queries to $H$, and we refer the reader to \zcref[S]{subsec:MR-game-based} for further details.

\paragraph{Encrypted key exchange: Simulation-based insecurity.} To generalize further to the setting of (O)EKE, we first abstract out an invertible ``cipher'' underlying our commitment and the \cite{masnyrindal} OT protocol. Given a bit $b$ and the public key $\pk$, the receiver applies the (randomized) function \[\cE: (b,\pk) \to (t_0,t_1) \qquad \text{where } t_{1-b} \gets \{0,1\}^n, t_b = \pk \oplus H(t_{1-b}).\] Although $\cE$ is randomized, there exists a deterministic inversion operation $\cD$ that given any choice of $b \in \{0,1\}$, returns the corresponding input $\pk$:  \[\cD : (b,(t_0,t_1)) \to \pk \qquad \text{where } \pk = t_b \oplus H(t_{1-b}).\]

In OEKE, a KEM public key $\pk$ is enciphered not under a single bit $b$, but under a (multi-bit) \emph{password} $\pw$. While the above choice of $(\cE,\cD)$ does not efficiently generalize to multi-bit passwords, researchers have considered other methods that do. Perhaps most simply, one can model $(\cE,\cD)$ as an ideal cipher, and have the PAKE initiator send $\phi = \cE(\pw,\pk)$ and the respondent recover $\pk = \cD(\pw,\phi)$. A simple and highly-efficient alternative that has been a focus of recent study instantiates $(\cE,\cD)$ as a 2-Feistel, which directly uses a random oracle $H$ as follows.
\[\cE(\pw,\pk) \to (T, r \oplus H(\pw,T)) \qquad \text{where } r \gets \{0,1\}^n, T = H(\pw,r) \oplus \pk.\] Notice that, given $\pw$, $\cD$ can recover the randomness $r$ used by $\cE$, which then allows it to recover $\pk$. And as mentioned in the introduction, there are other tweaks to $(\cE,\cD)$ that have been proposed in the literature, for instance the Tempo protocol \cite{arriaga2026tempo}.

(Un)fortunately, our attack on extractability is essentially agnostic to the choice of $(\cE,\cD)$ utilized by the protocol (as long as $\cD$ can recover any randomness sampled by $\cE$), and we show that for all known variants, there does not exist a simulator that can extract the choice of $\pw$ from a malicious initiator, invalidating claims of post-quantum UC-security. Similar to our negative results for Masny-Rindal OT, our impossibility here holds for \emph{any} choice of KEM: While the impossibility of extraction requires the KEM to have pseudorandom public keys, this assumption is already required to obtain security against offline password guessing attacks.

\paragraph{Encrypted key exchange: Game-based security.} We next study the post-quantum \emph{game-based} security of OEKE, which has thus far remained an open question, even when instantiated with a post-quantum KEM. In this work, we focus on the 2-Feistel variant, and consider security in the quantum random oracle model. We separately model security against malicious initiators and malicious respondents and in this overview, we'll focus on the more challenging case of malicious initiators. The security definition we target is as follows. Suppose the honest respondent samples their password $\pw \gets D$ from some distribution $D$ with min-entropy $\delta$. Then the best that we can hope for is that any malicious initiator, after interacting with the respondent, can distinguish the respondent's key from a uniformly random string with advantage at most $2^{-\delta} + \mathsf{negl}(\lambda)$. This captures the intuition that the best an adversary can do is make a single password guess $\pw'$ and, if $\pw' \neq \pw$, they will end up with $\mathsf{negl}(\lambda)$ information about the honest party's key. 

We begin by applying a similar measure-and-reprogram technique as above (though this requires extra care due to the structure of the 2-Feistel) to establish the following claim: Suppose a malicious initiator after sending its first message $\phi$ is challenged with two different respondent messages $c_0,c_1$ encrypted using two different passwords $\pw_0,\pw_1$ and yielding respondent keys $K_0,K_1$. Then, the probability that the initiator can output \emph{both} $K_0,K_1$ is $\mathsf{negl}(\lambda)$.

In order to translate this $\mathsf{negl}(\lambda)$ search security into a $2^{-\delta} + \mathsf{negl}(\lambda)$ distinguishing bound, we apply two quantum information techniques:

\begin{itemize}
    \item We derive a generalization of the mapping-distinguishing bound from \cite{dall2023necessity} that allows us to reason about the following scenario. Suppose that after $\phi$ is sent by the malicious initiator, a password $\pw \gets D$ is sampled and used to produce respondent message $c$ and output $K$. By sampling \emph{two} independent passwords $\pw_0,\pw_1$ from $D$ and then \emph{rewinding} the initiator, we can turn any initiator that predicts $K$ with better than $2^{-\delta} + \mathsf{negl}(\lambda)$ probability into one that can find both $K_0$ and $K_1$ with better than $\mathsf{negl}(\lambda)$ probability. 
    \item We next have to translate the search bound of $2^{-\delta} + \mathsf{negl}(\lambda)$ into (ideally) the \emph{same} distinguishing bound $2^{-\delta} + \mathsf{negl}(\lambda)$. In order to do so, we apply our advantage-tight one-way to hiding theorem (\zcref[S]{thm:search-to-decision}), which builds on the ``measure-rewind-measure'' technique of \cite{EC:KSSSS20} and quantum rewinding \cite{GSLV}. We defer a more in-depth overview of this theorem to \zcref[S]{sec:oth}.
    
\end{itemize}
\subsection{Discussion}

Broadly speaking, our work adds to the growing list of methods that are known to separate the quantum and classical adversarial models, even under post-quantum assumptions, following e.g. \cite{EC:YamZha21,FOCS:YamZha22,LombardiMQW22}. A major contribution of this work is that we did not design a (contrived) protocol, but showed that the security of practically-relevant PAKE and OT proposals fails to translate to the post-quantum setting.

For example, one could likely construct a contrived commitment scheme that is extractable against classical adversaries but not against quantum adversaries by incorporating a proof of quantumness protocol (e.g. \cite{BrakerskiCMVV18,FOCS:YamZha22}) into the design of the commitment. However, we show that there is \emph{already} a natural and widely used commitment scheme admitting such a separation. 

We believe that this adds even more weight to the argument that proofs of security in the \emph{classical} random oracle model (even from post-quantum assumptions) have limited meaning against quantum adversaries, and that the quantum adversarial model must be taken seriously when designing protocols based on post-quantum assumptions. In particular, it is crucial to gain a better understanding of the post-quantum security of (O)EKE as the community continues with standardization efforts.

\subsection{Future work}\label{sec:future}

Our work raises several directions for further research, a few of which we describe here.

\begin{itemize}
    \item \textbf{Consequence of our separation.} While we exhibit a formal breakdown of simulation-based security for (O)EKE and Masny-Rindal OT, our attacks don't have immediately ``obvious'' consequences. Indeed, we show that the protocols do retain some notions of game-based security. It would be interesting to find examples of broader contexts in which using (O)EKE or Masny-Rindal OT as a building block results in clear security issues stemming from the inability to extract the PAKE or OT input from the adversary. 
    \item \textbf{Alternative constructions.} Is there a practically-efficient and truly post-quantum alternative to the ``query-order'' extraction technique underlying Masny-Rindal OT and (O)EKE? In fact, can we come up with \emph{any} practically-efficient compiler from KEM to (two-message) OT or PAKE in the quantum random oracle model that achieves simulation-based security?
    \item \textbf{Post-quantum BPR security.} While we establish stand-alone game-based security for OEKE, the gold standard for game-based PAKE security considers a man-in-the-middle attacker and the fully \emph{concurrent} setting, following the definitions of \cite{EC:BelPoiRog00}. We leave the full BPR security of (2-Feistel-based) OEKE as an important direction for future work. 
    \item \textbf{Security in the quantum ideal cipher model.} One can also ask about the post-quantum game-based security of (O)EKE when instantiated with an \emph{ideal cipher}, especially given recent techniques for analyzing quantum-accessible random permutations (e.g. \cite{CojocaruHLYY25,carolan2026compressed,carolan2026compressedpermutationoraclesrevisited}). 
\end{itemize}

\section{Preliminaries}
\subsection{Quantum information}

We denote quantum registers $\Ar$ with sans-serif font. A quantum state $\rho_\Ar$ on register $\Ar$ is a positive semi-definite operator with trace 1. The reduced state $\rho_\Ar = \Tr_\Br(\rho_{\Ar\Br})$, where $\Tr(\cdot)$ is the partial trace operator. We denote pure states as $\ket{\psi}$ with corresponding density operator $\ketbra{\psi}$. We let $I$ denote the identity matrix. The trace distance between two state $\rho,\sigma$ is defined as \[\mathsf{TD}(\rho,\sigma) = \frac{1}{2}\| \rho - \sigma\|_1.\]

\paragraph{Quantum information lemmas.} We first state the Gentle Measurement lemma.

\begin{lemma}\label{lemma:gentle-measurement}
Let $\rho$ be a quantum state and let $(\Pi,I-\Pi)$ be a projective measurement such that $\Tr(\Pi\rho) \geq 1-\delta$. Let \[\rho' = \frac{\Pi\rho\Pi}{\Tr(\Pi\rho)}\] be the state after applying $(\Pi,I-\Pi)$ to $\rho$ and post-selecting on obtaining the first outcome. Then, $\mathsf{TD}(\rho,\rho') \leq 2\sqrt{\delta}$.
\end{lemma}

\noindent We will need the Cauchy-Schwarz inequality with respect to the Hilbert-Schmidt inner product $\langle A,B \rangle = \Tr(A^\dagger B)$.

\begin{lemma}
If $A,B$ are complex matrices s.t $A^\dagger B$ is defined,
     \[
     |\Tr(A^\dagger B)|^2 \leq \Tr(A^\dagger A)\Tr(B^\dagger B).
     \] \label{cshs}
\end{lemma}
\noindent We use the following quantum rewinding statement, which says if two different computations (on the same state) give specific outputs with sufficiently high probability, performing the computations sequentially obtains both outputs with noticeable probability.
\begin{lemma}[{\cite[Lemma 36]{dall2023necessity}}]\label{lemma:binrewind}
    Let $P,Q$ be projectors and $\rho$ any quantum state, with $\epsilon = \Tr(P\rho) + \Tr(Q\rho) - 1 \geq 0$. Then \[
    \frac{\epsilon^2}{4} \leq \Tr(PQP\rho).
    \]
\end{lemma}

\noindent We also need our own generalization of the above to $n$ different computations chosen according to a distribution $D$.

\begin{lemma}\label{multirewind}
    Let $n \geq 2$, let $D$ be a (non-point-mass) distribution on $[n]$ with min-entropy $\gamma$, define $D^2$ to be the distribution on $[n] \times [n]$ defined by sampling $i \gets D$ and then $j \gets D$ conditioned on $j \neq i$, let $\rho$ be any state, and let $P_1,\dots,P_n$ be projectors with \[\E_{i \gets D}[\Tr(P_i\rho)] = 2^{-\gamma} + \epsilon\] for $\epsilon \geq 0$. Then \[\E_{(i,j) \gets D^2}[\Tr(P_iP_jP_i\rho)] \geq \epsilon^3.\]
\end{lemma}

\begin{proof}
    We follow the outline of the proof of \zcref[S]{lemma:binrewind} in \cite{dall2023necessity} and adapt it to our setting. Define $p_i \coloneqq \Pr_{j \gets D}[j = i]$ and \[s \coloneqq \E_{i \gets D}[\Tr(P_i\rho)] = \sum_i p_i\Tr(P_i\rho) = 2^{-\gamma} + \epsilon.\]
    By applying \zcref[S]{cshs} with $A^\dagger = \sum_i p_i P_i\sqrt{\rho}$, $B = \sqrt{\rho}$ and using  $\Tr(\sqrt{\rho}^\dagger \sqrt{\rho}) = \Tr(\rho) = 1$, we have that \begin{align*}
s^2 &\leq \Tr\left(\left(\sum_i p_iP_i\right)\rho\left(\sum_i p_iP_i\right)\right) \\
&= \sum_{i,j}p_i p_j \Tr\left(P_i\rho P_j\right) \\
&= \sum_i p_i^2 \Tr(P_i\rho P_i) + \sum_{i \neq j}p_i p_j\Tr(P_i \rho P_j).
\end{align*} Since $P_i$ is a projector, $\Tr$ is cyclic, and $p_i \leq 2^{-\gamma}$ for every $i$, \[\sum_i p_i^2 \Tr\left(P_i\rho P_i\right) = \sum_i p_i^2\Tr\left(P_i\rho\right) \leq 2^{-\gamma}s,\] so rearranging gives
    \[s\epsilon = s^2 - 2^{-\gamma}s \leq \sum_{i \neq j} p_i p_j\Tr(P_i\rho P_j) \leq \sum_{i \neq j} p_i p_j\sqrt{\Tr(P_i P_jP_i\rho)\Tr(P_j\rho)},\] where the second inequality follows by taking $A^\dagger = P_j \sqrt{\rho}$ and $B = \sqrt{\rho} P_i P_j$ in \zcref[S]{cshs}. Applying Cauchy-Schwarz again, we obtain
    \begin{align*}
        s^2 \epsilon^2 &\leq \left(\sum_{i \neq j}p_i p_j\Tr(P_iP_jP_i\rho)\right)\left(
        \sum_{i \neq j}p_ip_j\Tr(P_j\rho)\right) \\ &\leq s\left(\sum_{i \neq j}p_i p_j\Tr(P_iP_jP_i\rho)\right),
    \end{align*} 
    so we can conclude that \[\sum_{i \neq j}p_i p_j\Tr(P_iP_jP_i\rho) \geq s\epsilon^2 \geq \epsilon^3.\] Finally, by definition of $D^2$,
    \begin{align*}
        \E_{(i,j) \gets D^2}[\Tr(P_i P_j P_i \rho)] &= \sum_{i \neq j} \frac{p_i p_j}{1-p_i}\Tr(P_i P_j P_i\rho) \\ &\geq \sum_{i \neq j}p_i p_j \Tr(P_i P_j P_i \rho) \\ &\geq \epsilon^3.
    \end{align*}
\end{proof}

\paragraph{Measure-and-reprogram.} Recall the well-known measure-and-reprogram technique introduced in \cite{C:DFMS19} and expanded in \cite{C:DonFehMaj20}. The idea is as follows: suppose $\adv^H$ makes quantum queries to a random oracle $H$ and outputs $x$ and (possibly quantum) $z$ such that $V(x,H(x),z) = 1$ where $V$ is some predicate. Then there is a (black box) simulator $\Sim$ such that, given access to $\adv$, and a random $\Theta$, measures a random $H$ query of $\adv$ to get $x$ and reprograms $H(x) = \Theta$. The guarantee is that with noticeable probability, $\Sim^\adv$ will output $x,z$ such that $V(x,\Theta,z) = 1$. For reprogramming $t$ queries, the simulator $\Sim$ works in multiple stages: it first outputs a permutation $\pi$ of $\{0, \dots, t-1\}$ and $x_{\pi(0)}$, and takes an input $\Theta_{\pi(0)}$. It then runs, and outputs $x_{\pi(1)}$ and takes as input $\Theta_{\pi(1)}$, etc., until it outputs $(\pi, \pi(\mathbf{x}), z)$. Formally:
\begin{theorem}\label{thm:mandp}
    There exists a black-box polynomial time $(t+1)$-stage quantum algorithm $\Sim$ with syntax as already outlined, with the following property. Let $\adv$ be a quantum oracle algorithm that makes $q$ queries to a random oracle $H: \mathcal{X} \to \mathcal{Y}$ and that outputs $(\mathbf{x} = (x_0, \dots, x_{t-1}), z)$. Then for any fixed $\mathbf{x}^*$ with distinct entries and (possibly quantum/randomized) predicate $V$,
    \begin{align*}
        \Pr_{\Theta}[\mathbf{x} = \mathbf{x}^* \land V^{H(\mathbf{x} \ast \mathbf{\Theta})}(\mathbf{x},\mathbf{\Theta},z) = 1 : (\pi, \pi(\mathbf{x}), z) \gets \langle \Sim^\adv, \pi(\mathbf{\Theta})\rangle] \\
        \geq \frac{1}{(2q+1)^{2t}}\Pr_{H}[\mathbf{x} = \mathbf{x}^* \land V^H(\mathbf{x},H(\mathbf{x}),z) = 1 : (\mathbf{x}, z) \gets \adv^H]
    \end{align*}
    Here $\mathbf{\Theta} = (\Theta_0, \dots, \Theta_{t-1})$ are chosen uniformly at random from $\mathcal{Y}$.
\end{theorem}
\noindent\zcref[S]{thm:mandp} is identical to \cite[Theorem~6]{C:DonFehMaj20} except that in our version $V$ has access to $H$ in the original game and the reprogrammed oracle $H(\mathbf{x} \ast \mathbf{\Theta})$ in the simulated game. The proof is very similar to \cite[Theorem~6]{C:DonFehMaj20} and is sketched in \zcref[S]{section:newmandrp}.\footnote{The coherent measure-and-reprogram theorem proved in \cite{AC:CGLS24} already allows $V$ to access $H$, but we opt not to use it, as the structure of \zcref[S]{thm:mandp} is more suited to our needs.}
\\
\\
We also use a result of \cite{EC:DFMS22}: Given an algorithm $\cA$ with access to a random oracle $H: \mathcal{X} \to \mathcal{Y}$ that eventually outputs a value $t = H(x)$ for some unknown $x$, there is a simulator $\Sim$ that simulates access to $H$ using an interface $\Sim.RO$, and allows extraction of $x$ given $t$ via an interface $x = \Sim.E(t)$.
\begin{theorem}[Quantum random oracle simulation with extraction {\cite{EC:DFMS22}}]\label{thm:simextract}
    There is an efficient quantum algorithm $\Sim$ with two interfaces $\Sim.RO, \Sim.E$ such that for any algorithm $\cA^H$ with quantum access to a random oracle $H: \mathcal{X} \to \mathcal{Y}$, we have the following guarantees;
    \begin{enumerate}
        \item \cite[Theorem~4.3~1]{EC:DFMS22} If no $\Sim.E$ query is made, $\Sim.RO$ is perfectly indistinguishable from $H$. \label{ext-indist}
        \item \cite[Theorem~4.3~2(c)]{EC:DFMS22} Any two independent queries to $\Sim.RO, \Sim.E$ $8\sqrt{2/|\mathcal{Y}|}$-almost commute (independent means they could be made in either order). \label{ext-commute}
        \item \cite[Proposition~4.5]{EC:DFMS22} Suppose $\cA$ makes $q$ queries to $H$. Then \[
        \Pr_{\substack{x,t \gets \cA^{\Sim.RO} \\ h \gets \Sim.RO(x) \\ x' \gets \Sim.E(t) }}[x' \neq x \land h = t] \leq O\left(\frac{q^3}{|\mathcal{Y}|}\right).
        \] \label{ext-correct}
    \end{enumerate}
\end{theorem}
\subsection{KEM}\label{subsec:KEM}
A key encapsulation mechanism \[\KEM = (\KG, \Enc, \Dec)\] consists of three polynomial time algorithms:
\begin{itemize}
    \item $\KG(1^\lambda;r) \to \pk,\sk$ is a randomized algorithm that generates a public key $\pk \in \PK_\lambda$ and a secret key $\sk \in \SK_\lambda$. It takes an input the security parameter $1^\lambda$ and we will sometimes be explicit about the random coins $r \in \R_\lambda$ that it uses. We will sometimes denote its output by $(\pk(r), 
    \sk(r)) = \KG(r)$, leaving the security parameter implicit.
    \item $\Enc(\pk) \to c,K$ is a randomized algorithm that, given a public key $\pk$, generates a ciphertext $c$ and output key $K$.
    \item $\Dec(\sk,c) = K$ is a deterministic algorithm that, given a ciphertext $c$ and secret key $\sk$, generates an output key $K$.
\end{itemize}

In what follows, we will often leave the security parameter $\lambda$ implicit, for example when writing the sets $\PK,\SK$, and $\R$. $P[i]$ denotes the (zero-indexed) $i$-th element of the tuple $P$.


\begin{definition}[KEM: Correctness]\label{def:KEM-correctness}
    A KEM $\KEM$ has correctness error $\delta(\lambda)$ if
\[
    \Pr\left[\Dec(\sk,c) \neq K : \begin{array}{r}(\pk, \sk) \gets \KG(1^\lambda) \\ (c,K) \gets \Enc(\pk) \end{array}\right] \leq \delta(\lambda).
\] We say that the KEM is \emph{correct} if $\delta(\lambda) = \mathsf{negl}(\lambda)$.
\end{definition}

We next define several notions of security for KEM that will be used in this work.

\paragraph{Public-key uniformity.} The post-quantum computational public key uniformity experiment for $\KEM$ consists of the following security game played against a quantum polynomial time adversary $\adv$.\\

\noindent\underline{$\text{PKU}_{\KEM, \lambda}^\adv$}

\begin{itemize}
    \item $\pk_0 \gets \PK$
    \item $(\pk_1,*) = \KG(1^\lambda, r); r \gets \R$
    \item $c \gets \{0,1\}$
    \item $c' \gets \adv(1^\lambda, \pk_c)$
\end{itemize}

Define the advantage of $\adv$ in the above game as \[
\mathsf{Adv}(\text{PKU}_{\KEM, \lambda}^\adv) = |\Pr[c' = 1 | c = 0] - \Pr[c' = 1 | c = 1]|.
\]
\begin{definition}[KEM: Public key uniformity]\label{def:KEM-unikeys}
    A KEM $\KEM$ has post-quantum computationally uniform public keys if for any QPT adversary $\adv$,
    \[ \mathsf{Adv}(\emph{PKU}_{\KEM, \lambda}^\adv) \leq \mathsf{negl}(\lambda).
\]
\end{definition}
\noindent By purifying the randomness of the challenger, we arrive at an equivalent definition of the security game that will be more convenient for us to work with.\\

\noindent\underline{$\text{PKU}_{\KEM, \lambda}^\adv$}
\begin{itemize}
    \item $\ket{\mu_0} = \frac{1}{\sqrt{|\PK|}}\sum_{\pk \in \PK} \ket{\pk}_{\Xr_0}\ket{\pk}_{\Yr_0}$
    \item $\ket{\mu_1} = \frac{1}{\sqrt{|\R|}}\sum_{r \in \R} \ket{r}_{\Xr_1}\ket{\pk(r)}_{\Yr_1}$
    \item $c \gets \{0,1\}$
    \item $c' \gets \adv(1^\lambda,\Yr_c)$
\end{itemize}

\paragraph{Key unpredictability.} Consider the following experiment.\\

\noindent\underline{$\text{UNP}_{\KEM, \lambda}^\adv$}
\begin{itemize}
    \item $(\pk, \sk) \gets \KG(1^\lambda)$
    \item $(c,K) \gets \Enc(\pk)$
    \item $K' \gets \adv(1^\lambda, \pk,c)$
\end{itemize}
Define the advantage of $\adv$ in the above game as \[
\mathsf{Adv}(\text{UNP}_{\KEM, \lambda}^\adv) = \Pr[K' = K] .
\]
\begin{definition}[KEM: Unpredictability Security]\label{def:KEM-unp} A KEM $\KEM$ has post-quantum key unpredictability if for any QPT adversary $\adv$,
    \[ \mathsf{Adv}(\emph{UNP}_{\KEM, \lambda}^\adv) \leq \mathsf{negl}(\lambda).
\]
\end{definition}

\paragraph{Key indistinguishability.} Consider the following experiment.\\

\noindent\underline{$\text{IND}_{\KEM, \lambda}^\adv$}
\begin{itemize}
    \item $(\sk, \pk) \gets \KG(1^\lambda)$
    \item $(c,K_0) \gets \Enc(\pk)$
    \item $K_1 \gets \{0,1\}^\lambda$
    \item $b \gets \{0,1\}$
    \item $b' \gets \adv(1^\lambda, \pk,c,K_b)$
\end{itemize}
Define the advantage of $\adv$ in the above game as \[
\mathsf{Adv}(\text{IND}_{\KEM, \lambda}^\adv) = |\Pr[b' = 1 | b = 0] - \Pr[b' = 1 | b = 1]|.
\]
\begin{definition}[KEM: Indistinguishability Security]\label{def:KEM-ind}
    A KEM $\KEM$ has post-quantum key indistinguishability if for any QPT adversary $\adv$,
    \[ \mathsf{Adv}(\emph{IND}_{\KEM, \lambda}^\adv) \leq \mathsf{negl}(\lambda).
\]
\end{definition}

\subsection{Masny-Rindal OT}

\subsubsection{Syntax}\label{subsubsec:MR-syntax}

Masny-Rindal (MR) OT \cite{masnyrindal} is a compiler from KEM to oblivious transfer in the random oracle model. Let $\KEM = (\KG, \Enc, \Dec)$, and let $H : \PK \to \PK$, where $\PK$ is the public key space of $\KEM$. Assume for simplicity that elements of $\PK$ are naturally specified as bitstrings so that $\oplus$ is an additive operation over the group $\PK$. The OT protocol is specified as follows.

\begin{itemize}
    \item $\Rec_1(1^\lambda,b)$: 
    \begin{itemize}
        \item Sample $t_{1-b} \gets \PK$, and $(\pk_b,\sk_b) \gets \KG(1^\lambda)$.
        \item Set $t_b \coloneqq \pk_b \oplus H(t_{1-b})$.
        \item Send $\msg_R \coloneqq (t_0,t_1)$ to the sender and keep $\sk \coloneqq (b,\sk_b)$.
    \end{itemize}
    \item $\Sen(\msg_R)$:
    \begin{itemize}
        \item Set $\pk_0 \coloneqq t_0 \oplus H(t_1)$ and $\pk_1 \coloneqq t_1 \oplus H(t_0)$.
        \item Sample $c_0,K_0 \gets \Enc(\pk_0)$ and $c_1,K_1 \gets \Enc(\pk_1)$.
        \item Send $\msg_S \coloneqq (c_0,c_1)$ to the receiver and output $m_0 \coloneqq K_0, m_1 \coloneqq K_1$.
    \end{itemize}
    \item $\Rec_2(\sk,\msg_S)$:
    \begin{itemize}
        \item Output $m = \Dec(c_b,\sk_b)$.
    \end{itemize}
\end{itemize}

\subsubsection{Game-based security}

We formulate a notion of game-based security for two-message oblivious transfer.

\begin{definition}[Two-message OT: Correctness]\label{def:OT-correctness}A two-message OT protocol $(\Rec_1,\Sen,\Rec_2)$ satisfies \emph{correctness} if for any $b \in \{0,1\}$,
    \[\Pr\left[m = m_b : \begin{array}{r} \sk,\msg_R \gets \Rec_1(1^\lambda,b) \\ (m_0,m_1),\msg_S \gets \Sen(\msg_R) \\ m \gets \Rec_2(\sk,\msg_S)\end{array}\right] = 1-\mathsf{negl}(\lambda).\]
\end{definition}

\begin{definition}[Two-message OT: Security against malicious receiver]\label{def:OT-sender}
    A two-message OT protocol $(\Rec_1,\Sen,\Rec_2)$ satisfies \emph{security against malicious receiver} if for any QPT $\cA$, 
    \[\Pr\left[m^* = m_b : \begin{array}{r} \msg_R,\sigma \gets \cA(1^\lambda) \\ (m_0,m_1),\msg_S \gets \Sen(\msg_R) \\ b \gets \{0,1\} \\
    m^* \gets \cA(\msg_S,b,\sigma)\end{array}\right] \leq  \frac{1}{2} + \mathsf{negl}(\lambda).\]
\end{definition}

\begin{definition}[Two-message OT: Security against malicious sender]\label{def:OT-receiver}
    A two-message OT protocol $(\Rec_1,\Sen,\Rec_2)$ satisfies \emph{security against malicious sender} if for any QPT $\cA$,
    \[\Bigg|\Pr\left[b^* = b : \begin{array}{r}b \gets \{0,1\} \\ \sk,\msg_R \gets \Rec_1(1^\lambda,b) \\ b^* \gets \cA(\msg_R)\end{array}\right] - \frac{1}{2}\Bigg| = \mathsf{negl}(\lambda).\]
\end{definition}

\subsubsection{Simulation-based security}

\cite{masnyrindal} define the following ideal functionality $\cF_\OT$, which they call \emph{endemic} security.\\

\noindent \underline{$\cF_\OT$:}
\begin{itemize}
    \item Take an input $b \in \{0,1\}$ from the receiver.
    \item If the sender is corrupt, take an input $m_0,m_1$ from the sender. If the receiver is corrupt, take an additional input $m_b$ from the receiver.
    \item All undefined $m$ are sampled uniformly from $\{0,1\}^\lambda$.
    \item Output $m_b$ to the receiver and $(m_0,m_1)$ to the sender.
\end{itemize}

\begin{definition}[Two-message OT: Simulation security]\label{OT:simulation}
    A two-message OT protocol $(\Rec_1,\Sen,\Rec_2)$ is simulation-secure against a malicious receiver if for any family of states $\ket{\psi}_{D,A} = \{\ket{\psi_\lambda}\}_{\lambda}$ on two registers $D,A$ and any QPT adversary $\cA$,  there exists a QPT simulator $\Sim$ such that for any QPT distinguisher $\cD$,
    \[\left|\Pr[\cD(D,\langle \cA(A),\Sen \rangle) = 1] - \Pr[\cD(D,\langle \Sim(A),\cF_\OT\rangle) = 1]\right| = \mathsf{negl}(\lambda),\] where 
    \begin{itemize}
        \item $\langle \cA(A),\Sen\rangle$ denotes the interaction between the adversary on input register $A$ and the honest sender, and outputs $\cA$'s output along with $\Sen$'s output $m_0,m_1$.
        \item $\langle \Sim(A),\cF_\OT\rangle$ denotes the interaction between the simulator on input register $A$ and the ideal functionality $\cF_\OT$. The output includes $\cA(A)$'s output along with the sender's part $(m_0,m_1)$ of $\cF_\OT$'s output.
    \end{itemize}

    The protocol is simulation-secure against a malicious sender if the analogous claim holds when the adversary corrupts the sender.
\end{definition}

\paragraph{Simulation of ideal functionalities.} One clarification is in order about the above definition. Since the protocol we will be analyzing involves an ideal functionality, namely a random oracle, we will allow the simulator to maintain the interface of the random oracle. That is, whenever the adversary or the distinguisher queries the random oracle, the simulator may process the query and respond. In the quantum setting, an ideal functionality is specified by a channel applied to input and output registers $X,Y$ supplied by the adversary. In this case, the simulator may take the registers $X,Y$, apply some operation (potentially also on its private register), and return $X,Y$.

\paragraph{Classical UC security.} In Section 4 and Appendix E, \cite{masnyrindal} show that the Masny-Rindal OT protocol is UC-secure against \emph{classical} adversaries when instantiated with any KEM that has negligible correctness error (\zcref[S]{def:KEM-correctness}), computationally uniform public keys (\zcref[S]{def:KEM-unikeys}), and key indistinguishability (\zcref[S]{def:KEM-ind}). As a consequence, this establishes that the protocol is stand-alone simulation-secure according to our above definition restricted to classical adversaries and simulators.

\subsection{OEKE}

\subsubsection{Syntax}\label{subsubsec:OEKE-syntax}

Once-encrypted key exchange (OEKE) is a template for building two-message PAKE specified by the following ingredients.
\begin{itemize}
    \item A KEM $\KEM = (\KG, \Enc, \Dec)$.
    \item (Potentially oracle-aided) encryption and decryption algorithms $\cE^{O_{\mathsf{Enc}}}, \cD^{O_{\mathsf{Enc}}}$ with the following syntax:
    \begin{itemize}
        \item $\phi \coloneqq \cE^{O_{\mathsf{Enc}}}(\pw,\pk;r)$ takes a password $\pw$, a KEM public key $\pk$, and random coins $r$, and outputs a message $\phi$.
        \item $\pk \coloneqq \cD^{O_{\mathsf{Enc}}}(\pw,\phi)$ takes a password $\pw$ and a message $\phi$ and outputs $\pk$.
    \end{itemize}
    \item A (potentially oracle-aided) key derivation function $\cK^{O_{\mathsf{Key}}}(K) = \pkey$.
    \item A (potentially oracle-aided) tag function $\cT^{O_{\mathsf{Tag}}}(\pw,\pk,\phi,c,K) = \tau$. 
\end{itemize}


\noindent Given the above ingredients, the OEKE protocol is specified as follows.
\begin{itemize}
    \item $\PAKE_1(\pw) \to (\st,\phi)$:
    \begin{itemize}
        \item Sample $(\pk,\sk) \gets \KG$.
        \item Sample $r$ and compute $\phi = \cE^{O_{\mathsf{Enc}}}(\pw,\pk;r)$.
        \item Set $\st = (\sk,\pk,\phi)$
    \end{itemize}
    \item $\PAKE_2(\pw,\phi) \to (c,\tau,\pkey)$:
    \begin{itemize}
        \item Compute $\pk = \cD^{O_{\mathsf{Enc}}}(\pw,\phi)$.
        \item Sample $(c,K) \gets \Enc(\pk)$.
        \item Set $\pkey = \cK^{O_{\mathsf{Key}}}(K)$.
        \item Set $\tau = \cT^{O_{\mathsf{Tag}}}(\pw,\pk,\phi,c,K)$.
    \end{itemize}
    \item $\PAKE_3(\pw,\st,c,\tau) \to \pkey$:
    \begin{itemize}
        \item Compute $K = \Dec(\sk,c)$.
        \item Set $\pkey = \cK^{O_{\mathsf{Key}}}(K)$.
        \item If $\tau = \cT^{O_{\mathsf{Tag}}}(\pw,\pk,\phi,c,K)$, output $\pkey$, otherwise output $\bot$.
    \end{itemize}
\end{itemize}
Besides correctness and security of $\KEM$, we also require public key uniformity (\zcref[S]{def:KEM-unikeys}) for post-quantum security; otherwise a quantum adversary can mount an offline password attack by obtaining $\phi = \ICEnc(\pw, \pk)$, computing $\pk_j = \ICDec(\pw_j, \phi)$ for many candidate passwords $\pw_j$, and checking which $\pk_j$ follows the public key distribution. 

\subsubsection{Game-based security}

Next, we formulate a game-based stand-alone notion of security for two-message PAKE.

\begin{definition}[Two-message PAKE: Correctness]\label{def:PAKE-correctness} A two-message PAKE scheme $(\PAKE_1,\PAKE_2,\PAKE_3)$ satisfies \emph{correctness} if for any $\pw$,
\[\Pr\left[\pkey_A = \pkey_B : \begin{array}{r} (\st,\phi) \gets \PAKE_1(\pw) \\ (c,\tau,\pkey_B) \gets \PAKE_2(\pw,\phi) \\ \pkey_A \gets \PAKE_3(\pw,\st,c,\tau)\end{array}\right] = 1-\mathsf{negl}(\lambda).\]
    
\end{definition}

\begin{definition}[Two-message PAKE: Passive security]\label{def:PAKE-passive} A two-message PAKE scheme $(\PAKE_1,\PAKE_2,\PAKE_3)$ satisfies \emph{passive security} if for any $\pw$ and QPT adversary $\cA$, 
\[\Pr\left[\cA(\pw,\phi,c,\tau,\chall) = b : \begin{array}{r} (\st,\phi) \gets \PAKE_1(\pw) \\ (c,\tau,\pkey) \gets \PAKE_2(\pw,\phi) \\ b \gets \{0,1\} \\ \text{If } b = 0, \chall \coloneqq \pkey \\ \text{If } b = 1, \chall \gets \{0,1\}^\lambda\end{array}\right] \leq \frac{1}{2} + \mathsf{negl}(\lambda). \]
    
\end{definition}
\begin{definition}[Two-message PAKE: Security against malicious initiator]\label{def:PAKE-initiator} A two-message PAKE scheme $(\PAKE_1,\PAKE_2,\PAKE_3)$ satisfies \emph{security against malicious initiator} if for any distribution $D$ with min-entropy $\gamma$ over passwords $\mathcal{PW}$ and any QPT adversary $(\cA_1, \cA_2)$, \[\Pr\left[\cA_2(\pw,c,\tau,\chall,\sigma) = b : \begin{array}{r} \pw \gets D \\  (\phi^*,\sigma) \gets \cA_1(1^\lambda) \\ (c,\tau,\pkey) \gets \PAKE_2(\pw,\phi^*) \\ b \gets\{0,1\} \\ \text{If } b= 0, \chall \coloneqq \pkey \\ \text{If } b = 1, \chall \gets \{0,1\}^\lambda\end{array}\right] \leq \frac{1}{2} + \frac{2^{-\gamma}}{2} + \mathsf{negl}(\lambda).\]
\end{definition}

\begin{definition}[Two-message PAKE: Security against malicious respondent]\label{def:PAKE-respondent} A two-message PAKE scheme $(\PAKE_1,\PAKE_2,\PAKE_3)$ satisfies \emph{security against malicious respondent} if for any distribution $D$ with min-entropy $\gamma$ over passwords $\mathcal{PW}$ and any QPT adversary $(\cA_1,\cA_2)$, \begin{align*}\Pr\left[\cA_2(\pw,\chall,\sigma) = b  : \begin{array}{r} \pw \gets D \\  (\st,\phi) \gets \PAKE_1(\pw) \\ (c^*,\tau^*,\sigma) \gets \cA_1(\phi) \\ \pkey \gets \PAKE_3(\pw,\st,c^*,\tau^*) \\ b \gets\{0,1\} \\ \text{If } \pkey = \bot, \chall \coloneqq \bot \\ \text{If } b= 0, \chall \coloneqq \pkey \\ \text{If } b = 1, \chall \gets \{0,1\}^\lambda\end{array}\right] \leq \frac{1}{2} + \frac{2^{-\gamma}}{2} + \mathsf{negl}(\lambda).\end{align*}
\end{definition}

\subsubsection{UC security}
Universally Composable security was introduced by Canetti \cite{FOCS:Canetti01} as a version of simulation-based security that provides security when the protocol is concurrently composed with other protocols. UC is similar to standalone security (as in \zcref[S]{OT:simulation}) except the distinguisher is now \emph{interactive} (the interactive distinguisher is referred to as the environment $\Env$).
\cite{EC:Unruh10} formally defined UC-security for quantum protocols and adversaries and proved that secure protocol composition still holds. We adopt the same framework, so the simulator, adversary and ideal functionalities can be quantum. However, we study post-quantum protocols, so all honest parties are classical and send/receive only classical messages, and the ideal functionality being UC-realized is classical.

\paragraph{A word on quantum ideal functionalities.} In order to meaningfully talk about the post-quantum security of protocols in the UC framework that use a quantum accessible random oracle/ideal cipher, one must define a suitable ideal functionality. To the best of  our knowledge, a formal definition for either has not appeared in the literature (\cite{EC:LyuLiuHan24} provides a UC proof in the QROM but makes no reference to an ideal functionality). A priori, it is not obvious that such ideal functionalities exist, since in the UC framework ideal functionalities must be \emph{efficient}, and defining efficient procedures for  simulating quantum access to ideal primitives with negligible error has proven quite challenging, with a long line of recent work \cite{C:Zhandry19,AC:Unruh23,carolan2026compressed,grinko2025quantum,foxman2026quantum}.

By \cite[Lemma~4]{C:Zhandry19} Zhandry's compressed standard oracle is an efficient and perfectly correct simulation of quantum access to a random oracle by the unitary $\ket{x,y} \mapsto \ket{x,y \oplus H(x)}$,  so it can be used to define a corresponding UC functionality for the QROM\footnote{One could also define an ideal functionality where access is given via the phase query interface $\ket{x,y} \mapsto (-1)^{y \cdot H(x)}\ket{x,y}$ but this is not useful in the post-quantum setting, as we assume classical algorithms querying the random oracle instantly measure the result in the computational basis.}. 

Ideal ciphers are more challenging: for example, there is no known quantum analogue of \cite{C:DaiSte16} enabling efficient simulation of an ideal cipher using a random oracle. It seems plausible that the techniques of \cite{carolan2026compressed} can be applied to an ideal cipher, but this would incur substantial simulation error (perhaps mitigated by very recent work \cite{carolan2026compressedpermutationoraclesrevisited}). Two other recent works \cite{grinko2025quantum,foxman2026quantum} define a procedure to perfectly simulate a broad class of oracles that includes ideal ciphers,\footnote{Technically, this is the class of oracles sampled from a compact Lie group according to the Haar measure.} but it is not yet known if this ideal cipher simulation is efficient\footnote{\cite{grinko2025quantum,foxman2026quantum} both give simulations of random permutations whose efficiency reduces to efficiency of the Clebsch-Gordan transform for permutations, which is currently an open problem.}. However, this is not an issue for us, as we aim to show UC-\emph{insecurity}. In other words, we show that even if such an ideal functionality existed, (O)EKE would be UC-insecure relative to it.

\begin{figure}[t]
    \framebox{\begin{minipage}{0.98\linewidth}
    \begin{itemize}
        \item On input $(\NewS, \sid, P_1, P_2, \pw, \role)$ from $P_1$, send $(\NewS, \sid, P_1, P_2, \role)$ to $\Sim$. Furthermore, if this is the first $\NewS$ message for $\sid$, or this is the second $\NewS$ message for $\sid$ and there is a record $\left<P_2, P_1, \cdot, \cdot \right>$, then record $\left<P_1, P_2, \pw, \role\right>$ and mark it $\fresh$.
        \item On $(\TestPwd, \sid, P_1, \pw^*)$ from $\Sim$, if there is a record $\left<P_1, P_2, \pw, \cdot\right>$ marked $\fresh$, then do:
            \begin{itemize}
              \item If $\pw^* = \pw$, then mark the record $\compromised$ and send ``correct guess'' to $\Sim$.
              \item If $\pw^* \neq \pw$, then mark the record $\interrupted$ and send ``wrong guess'' to $\Sim$.
            \end{itemize}
        \item On $(\NewKey, \sid, P_1, K^* \in \{0,1\}^\lambda)$ from $\Sim$, if there is a record $\left<P_1, P_2, \pw, \role \right>$, and this is the first $\NewKey$ message for $\sid$ and $P_1$, then output $(\sid, K)$ to $P_1$, where $K$ is defined as follows:
            \begin{itemize}
              \item If the record is $\compromised$, then set $K := K^*$.
              \item If the record is $\fresh$, a key $(\sid, K')$ has been output to $P_2$, at which time there was a record $\left<P_2, P_1, \pw, \cdot \right>$ marked $\fresh$, and $K' \neq \bot$, then set $K := K'$.
              \item Otherwise if $\role$ = ``respondent'' sample $K \gets \{0,1\}^\lambda$; if $\role$ = ``initiator'', set $K := \bot$.
            \end{itemize}
            Finally, mark the record $\completed$.
    \end{itemize}
    \end{minipage}}
        \caption{UC PAKE functionality $\FPake$}
    \label{fig:pake-functionality}
\end{figure}

\paragraph{The PAKE functionality.} Our UC PAKE functionality $\FPake$ is the same that was used in \cite{cryptoeprint:2026/1331} to analyze OEKE, and is given in \zcref[S]{fig:pake-functionality}: it is a minor modification of the standard PAKE functionality that incorporates one-sided explicit authentication. We now give a brief overview of how it works.

The functionality involves two parties, $P_1$ with password $\pw$ and $P_2$ with password $\pw'$. Each party is also given a $\role$: the party with $\role =$ ``initiator'' sends the first protocol message. The functionality captures explicit authentication for the ``initiator'' party. Each execution of the protocol involves a session for $P_1$ and a session for $P_2$, both of which have a corresponding record marked with a state:

\begin{itemize}
  \item When a session is established (by a $\NewS$ command), it is marked $\fresh$, and will remain $\fresh$ until the simulator (ideal adversary) attacks the session.
  \item The simulator may attack a $\fresh$ $P$ session by sending a $\TestPwd$ command for $P$, on a \emph{password guess} $\pw^*$. This models an \emph{online guessing attack} in the real world, where the simulator runs the algorithm of $P'$ on a password guess $\pw^*$ and communicates with $P$. If $\pw^* = \pw$, this is a successful attack, and the session is marked $\compromised$; otherwise the session is marked $\interrupted$. Note that once a session becomes $\compromised$ or $\interrupted$, it can never return to $\fresh$; this in particular means that $\TestPwd$ can be run only once on any specific session.
  \item The only way a party receives a session key is by the simulator sending a $\NewKey$ command, after which the session will output a key (or possibly the $\bot$ symbol for the initiator) to the party depending on the states of this session and its counter session. Subsequently the session is marked $\completed$ (so that $\TestPwd$ cannot be sent after the session ends).
\end{itemize}
We now emphasize a few points about $\FPake$ that will be important later.
\begin{enumerate}
    \item Notice that a simulator can only guess the password of a party once: after that the record is no longer $\fresh$. We define $\pw^*$ to be the password guess for a party $P$ in session $\sid$ if $\Sim$ sent the command $(\TestPwd, \sid, P, \pw^*)$ when there was a record $\langle P, P', \cdot, \cdot \rangle$ marked $\fresh$. If no such command has been sent, the password guess is $\bot$. Once $\NewKey$ has been sent to $P$ the password guess is fixed, as $\Sim$ can no longer send $\TestPwd$ to $P$.
    \item If $\Sim$ is to be a successful simulator, whenever a party $P$ outputs a key in the real world, $\Sim$ must send $\NewKey$ to $P$ in the ideal world, otherwise the real and ideal world are trivially distinguishable.
    \item When $\Sim$ sends $\NewKey$ to the respondent $P$, unless the record for $P$ has already been marked $\compromised$ (which requires $\Sim$ to have made a correct password guess), the key of $P$ is chosen uniformly and independently of $\Sim$'s view, and will remain independent of $\Sim$'s view for the rest of the security experiment. 
\end{enumerate}

\section{Attacks on simulation-based security}\label{sec:attack}

\subsection{Masny-Rindal OT}

\begin{theorem}\label{thm:MR-attack}
    Let $\Pi = (\Rec_1,\Sen,\Rec_2)$ be the Masny-Rindal OT protocol (\zcref[S]{subsubsec:MR-syntax}) instantiated with any KEM with negligible correctness error (\zcref[S]{def:KEM-correctness}). Then $\Pi$ is not post-quantum simulation-secure (\zcref[S]{OT:simulation}).
\end{theorem}

\begin{proof}

    First, we note that security against a malicious sender fails if the KEM does not have computationally uniform public keys (\zcref[S]{def:KEM-unikeys}). Indeed, given $t_0,t_1$ produced by the honest receiver on input bit $b$, the sender can recover $\pk_0 = t_0 \oplus H(t_1), \pk_1 = t_1 \oplus H(t_0)$. Since $t_{1-b}$ was sampled uniformly at random by the honest receiver, $\pk_{1-b}$ is distributed uniformly at random. On the other hand, if KEM does not have computationally uniform keys,  $\pk_b$ is \emph{distinguishable} from uniformly random, which gives the sender a non-negligible advantage in distinguishing $\pk_b$ from $\pk_{1-b}$ and thus in predicting $b$.

    Therefore, we assume that KEM has computationally uniform public keys, and to complete the proof we show that $\Pi$ is not simulation-secure against a malicious receiver. In particular, we define a (mixed) state\footnote{It is straightforward to obtain a counterexample using a pure state $\ket{\psi}_{\Dr,\Ar}$ by \emph{coherently} recording the choice of the branch $c \in \{0,1\}$ in register $\Dr$.} on registers $\Dr,\Ar$ and an adversary distinguisher pair $\cA,\cD$ for which no simulator $\Sim$ can possibly exist. 

    The state will be a uniform mixture over two pure states $\ket{\psi_0}_{\Dr,\Ar}$ and $\ket{\psi_1}_{\Dr,\Ar}$ specified by a choice of branch $c \in \{0,1\}$. In a slight abuse of notation, we consider the choice of the bit $c$ to be stored in register $\Dr$, so it is known to the distinguisher (but not necessarily to the adversary). These states are defined as



    \begin{align*}\ket{\psi_0} &= \frac{1}{\sqrt{2|\PK|^2}}\sum_{b \in \{0,1\},t_{1-b} \in \PK,\pk_b \in \PK}\ket{\pk_b}_{\Dr}\ket{b}_\Br\ket{t_{1-b}}_{\Xr}\ket{\pk_b}_\Kr, \\ \ket{\psi_1} &= \frac{1}{\sqrt{2|\PK||\R|}}\sum_{b \in \{0,1\},t_{1-b} \in \PK,r_b \in \R}\ket{r_b}_{\Dr}\ket{b}_\Br\ket{t_{1-b}}_{\Xr}\ket{\pk(r_b)}_\Kr,
    \end{align*}

    where $\Ar = (\Br,\Xr,\Kr)$. Now, $\cA$ takes as input register $\Ar$ and does the following.


    \begin{enumerate}
        \item Initialize a fresh register $\Yr$ and query $H$ on registers $\Xr,\Yr$ to obtain the following state:
        \begin{align*}
            &\text{if } c = 0: \quad \frac{1}{\sqrt{2|\PK|^2}}\sum_{b \in \{0,1\},t_{1-b} \in \PK,\pk_b \in \PK}\ket{\pk_b}_{\Dr}\ket{b}_\Br\ket{t_{1-b}}_\Xr\ket{\pk_b}_\Kr\ket{H(t_{1-b})}_{\Yr} \\
            &\text{if } c = 1: \quad \frac{1}{\sqrt{2|\PK||\R|}}\sum_{b \in \{0,1\},t_{1-b} \in \PK,r_b \in \R}\ket{r_b}_{\Dr}\ket{b}_\Br\ket{t_{1-b}}_\Xr\ket{\pk(r_b)}_\Kr\ket{H(t_{1-b})}_{\Yr}.
        \end{align*}
        \item Initialize fresh registers $\Treg_0,\Treg_1$ and apply the computation to registers $(\Br,\Xr,\Kr,\Yr,\Treg_0,\Treg_1)$ that produces the following state:
        \begin{align*}
            &\text{if } c = 0: \quad \frac{1}{\sqrt{2|\PK|^2}}\sum_{b \in \{0,1\},t_{1-b} \in \PK,\pk_b \in \PK}\ket{\pk_b}_{\Dr}\ket{b}_\Br\ket{t_{1-b}}_\Xr\ket{\pk_b}_\Kr\ket{H(t_{1-b})}_{\Yr}\ket{t_0}_{\Treg_0}\ket{t_1}_{\Treg_1} \\
            &\text{if } c = 1: \quad \frac{1}{\sqrt{2|\PK||\R|}}\sum_{b \in \{0,1\},t_{1-b} \in \PK,r_b \in \R}\ket{r_b}_{\Dr}\ket{b}_\Br\ket{t_{1-b}}_\Xr\ket{\pk(r_b)}_\Kr\ket{H(t_{1-b})}_{\Yr}\ket{t_0}_{\Treg_0}\ket{t_1}_{\Treg_1},
        \end{align*}
        where $t_b = \pk_b \oplus H(t_{1-b})$ in branch $c = 0$ and $t_b = \pk(r_b) \oplus H(t_{1-b})$ in branch $c=1$.
        \item Query $H$ on registers $\Xr,\Yr$, discard $\Yr$, and measure $\Treg_0,\Treg_1$ to obtain $(t_0,t_1)$. On branch $c=0$, this yields the state
        \[\frac{1}{\sqrt{2}}\sum_{b \in \{0,1\}}\ket{t_b \oplus H(t_{1-b})}_{\Dr}\ket{b}_\Br\ket{t_{1-b}}_\Xr\ket{t_b \oplus H(t_{1-b})}_\Kr.\] We leave the precise description of the resulting state on branch $c=1$ implicit, as it will not be important for the argument.
        
        \item Query $H$ (classically) on $t_0$ and $t_1$ to obtain $y_0$ and $y_1$.
        \item Output $(t_0,t_1)$ to the sender, receive $(c_0,c_1)$, and send $(\Br,\Xr,\Kr,(t_0, t_1, y_0, y_1, c_0, c_1))$ to the distinguisher $\cD$, which is a mix of quantum registers and classical information. 
    \end{enumerate}

    Finally, we describe the behavior of the distinguisher $\cD$, who knows the value of $c$, and additionally receives $(m_0,m_1)$ (from the honest sender in the real experiment, and from $\cF_\OT$ in the ideal experiment).
    \begin{itemize}
        \item If $c = 0$, define \[\ket{\psi_\chck} \coloneqq \frac{1}{\sqrt{2}}\sum_{b \in \{0,1\}}\ket{t_b \oplus y_{1-b}}_{\Dr}\ket{b}_\Br\ket{t_{1-b}}_\Xr\ket{t_b \oplus y_{1-b}}_\Kr\] and apply the measurement $\{\ketbra{\psi_\chck}, I - \ketbra{\psi_\chck}\}$ to registers $(\Dr,\Br,\Xr,\Kr)$. Output 1 if the first outcome and 0 if the second. 
        \item If $c = 1$, measure $\Dr,\Br$ to obtain $(r_b,b)$, compute $m = \Dec(\sk(r_b),c_b)$, and output 1 if $m = m_b$ and 0 otherwise.
    \end{itemize}

    Now, suppose towards contradiction that there exists a simulator $\Sim$ for which  \[\left|\Pr[\cD(\Dr,\langle \cA(A),\Sen \rangle) = 1] - \Pr[\cD(\Dr,\langle \Sim(\Ar),\cF_\OT\rangle) = 1]\right| = \mathsf{negl}(\lambda).\] Let \[p_\Real = \Pr[\cD(\Dr,\langle \cA(A),\Sen \rangle) = 1], \quad p_\Ideal = \Pr[\cD(\Dr,\langle \Sim(\Ar),\cF_\OT\rangle) = 1].\] Under this assumption, we can establish the following two claims, which yield a contradiction.
    \end{proof}


    \begin{claim}
        $p_\Real = 1-\mathsf{negl}(\lambda)$.
    \end{claim}

    \begin{proof}
        In branch $c=0$, the distinguisher outputs 1 with probability 1 due to the definition of the state on registers $(\Dr,\Br,\Xr,\Kr)$. In branch $c=1$, the correctness of the KEM directly implies that $\Pr[m = m_b] = 1-\mathsf{negl}(\lambda)$, and so the distinguisher outputs 1 with probability $1-\mathsf{negl}(\lambda)$.

    \end{proof}

    \begin{claim}
        $p_\Ideal \leq 3/4 + \mathsf{negl}(\lambda)$.
    \end{claim}

    \begin{proof}
        Consider the following sequence of hybrids, where each hybrid outputs a bit.

        \begin{itemize}
            \item $\Hyb_0$: This is the ideal game: $\cD(\Dr,\langle \Sim(\Ar),\cF_\OT\rangle)$.
            \item $\Hyb_1$: Change how we compute the output in branch $c=1$ as follows. Let $(b^*,m^*)$ be the input that $\Sim$ chooses to send to $\cF_\OT$, and let $b$ be the bit measured by the distinguisher. Output 1 when $b = b^*$ and 0 otherwise.

            We can argue $ \Pr[\Hyb_1 = 1] \geq  \Pr[\Hyb_0 = 1] - \mathsf{negl}(\lambda)$ as follows. First, immediately outputting 1 when $b = b^*$ instead of computing and comparing $m$ with $m_b$ only potentially increases the probability of outputting 1 when $b = b^*$. Next, when $b \neq b^*$, we have that $m_b$ is chosen by $\cF_\OT$ as a uniform string in $\{0,1\}^\lambda$, independent of the rest of the experiment, and thus $\Pr[m = m_b] = \mathsf{negl}(\lambda)$. 
            \item $\Hyb_2$: Replace the branch $c=1$ state with the branch $c=0$ state. That is, prepare the state  
            \[\frac{1}{\sqrt{2|\PK|^2}}\sum_{b \in \{0,1\},t_{1-b} \in \PK,\pk_b \in \PK}\ket{\pk_b}_{\Dr}\ket{b}_\Br\ket{t_{1-b}}_\Xr\ket{\pk_b}_\Kr\] regardless of whether $c=0$ or $c=1$. 


            $|\Pr[\Hyb_1 = 1] - \Pr[\Hyb_2 = 1]| = \mathsf{negl}(\lambda)$ follows directly from the computational public key uniformity of the KEM. Indeed, notice that neither $\Hyb_1$ nor $\Hyb_2$ need to operate on the $\Dr$ register in branch $c=1$ to compute its output bit. The distinguisher no longer uses $\Dr$, and $\Dr$ is inaccessible to $\Sim$. Thus, this register can be held by an external challenger in the computational public key uniformity game.
            \item $\Hyb_3$: Right before the distinguisher makes its branch-dependent measurement, post-select the registers $(\Dr,\Br,\Xr,\Kr)$ onto \[\ket{\psi_\chck} = \frac{1}{\sqrt{2}}\sum_{b \in \{0,1\}}\ket{t_b \oplus y_{1-b}}_{\Dr}\ket{b}_\Br\ket{t_{1-b}}_\Xr\ket{t_b \oplus y_{1-b}}_\Kr.\]

             $|\Pr[\Hyb_2 = 1] - \Pr[\Hyb_3 = 1]| = \mathsf{negl}(\lambda)$ follows from the assumption that $|p_\Real - p_\Ideal| = \mathsf{negl}(\lambda)$, the fact that $p_\Real = 1-\mathsf{negl}(\lambda)$, and Gentle Measurement applied to the measurement $\{\ketbra{\psi_\chck},I-\ketbra{\psi_\chck}\}$ performed in the $c = 0$ branch.      
        \end{itemize}

    Combining the above, we have that $ \Pr[\Hyb_3 = 1] \geq  \Pr[\Hyb_0 = 1] - \mathsf{negl}(\lambda)$. However, it is easy to see that $\Pr[\Hyb_3 = 1] = 3/4$. Indeed, with probability 1/2, we are in branch $c=1$, and in this branch, measuring register $\Br$ of the pure state $\ket{\psi_\chck}$ yields a bit that is uniform and independent of $b^*$, meaning that $\Pr[b = b^*] = 1/2$. Thus, $p_\Ideal \leq 3/4 + \mathsf{negl}(\lambda)$. 
        
    \end{proof}

\subsection{OEKE with ideal cipher}\label{subsec:oekeattack}


We first consider OEKE (\zcref[S]{subsubsec:OEKE-syntax}) where the encryption and decryption algorithms are instantiated with an ideal cipher. That is, let $(\IC,\IC^{-1})$ be an ideal cipher, and let $\cE^{(\IC,\IC^{-1})}(\pw,\pk) = \IC(\pw,\pk)$ and $\cD^{(\IC,\IC^{-1})}(\pw,\phi) = \IC^{-1}(\pw,\phi)$.

\begin{theorem}\label{thm:IC-attack}
    Let $\Pi$ be an OEKE protocol instantiated with (1) any KEM with negligible correctness error (\zcref[S]{def:KEM-correctness}), (2) the IC-based encryption and decryption algorithms $\cE,\cD$ defined above, (3) any key derivation function $\cK$, and (4) any tag function $\cT$. Then $\Pi$ does not post-quantum UC-realize $\FPake$\footnote{By \cite{PKC:RoyXu23}, UC-realizing the UC PAKE functionality does not imply correctness.} (\zcref[S]{fig:pake-functionality}).
\end{theorem}

\begin{proof}
By the same argument made at the beginning of the proof of \zcref[S]{thm:MR-attack} (see also the discussion in \zcref[S]{subsubsec:OEKE-syntax}), we can assume without loss of generality that $\KEM$ has post-quantum computationally uniform public keys (\zcref[S]{def:KEM-unikeys}), as otherwise $\Pi$ would not satisfy security against a malicious respondent.

Now, suppose $\Sim$ is a simulator for the protocol in \zcref[S]{OEKEdiag} with respect to the dummy adversary. Fix $\pw_0 \neq \pw_1$ to be any two distinct passwords. We construct an environment $\Env$ that distinguishes between the real and ideal world.
\paragraph{Environment $\Env_0$}
\begin{enumerate}
    \item Sample $b \gets \{0,1\}$ and send $(\NewS, \sid, P_2, P_1, \pw_b, \mathsf{respondent})$ to $P_2$ for some nonce $\sid$.
    \item Prepare the normalized state \[\frac{1}{\sqrt{2|\PK|}}\sum_{u \in \{0,1\}, \pk \in \PK} \ket{\pw_u}_{\Ar}\ket{\pw_u}_{\Xr}\ket{\pk}_{\Br}\ket{\pk}_{\Yr}.\]
    \item Query $\ICEnc$ on registers $\Xr, \Yr$ into a new register $\Zr$ (through $\adv$).
    \item Measure $\Zr$ to get $\phi$ and leftover state $\rho^0_{\Ar \Br \Xr \Yr}$.
    \item Query $\ICDec(\pw_u, \phi) = \pk_u$ for $u \in \{0,1\}$.
    \item Let $\ket{\psi} = \frac{1}{\sqrt{2}} \sum_{u \in \{0,1\}} \ket{\pw_u}_{\Ar}\ket{\pw_u}_{\Xr}\ket{\pk_u}_{\Br}\ket{\pk_u}_{\Yr}$ and apply the measurement $M = \{\Pi_\psi, I - \Pi_\psi\}$ to $\rho^0$ (here $\Pi_\psi$ is the projector onto $\ketbra{\psi}$).
    \item Output $1$ if $M$ accepts and otherwise output $0$.
\end{enumerate}
\paragraph{Environment $\Env_1$}
\begin{enumerate}
    \item Sample $b \gets \{0,1\}$ and send $(\NewS, \sid, P_2, P_1, \pw_b,\mathsf{respondent})$ to $P_2$ for some nonce $\sid$.
    \item Prepare the normalized state $\frac{1}{\sqrt{2|\R|}} \sum_{u \in \{0,1\}, r \in \R} \ket{\pw_u}_{\Ar}\ket{\pw_u}_{\Xr}\ket{r}_{\Br}\ket{\pk(r)}_{\Yr}$ using $\KG$.
    \item Query $\ICEnc$ on registers $\Xr, \Yr$ into a new register $\Zr$ (through $\adv$).
    \item Measure $\Zr$ to get $\phi$ and leftover state $\rho^1_{\Ar \Br \Xr \Yr}$.
    \item Query $\ICDec(\pw_u, \phi) = \pk_u$ for $u \in \{0,1\}$.
    \item Measure $\Ar, \Br$ to get $\pw_i, r_i$; if $\pw_i \neq \pw_b$ output $1$.
    \item Instruct $\adv$ to send $\phi$ to $P_2$.
    \begin{itemize}
        \item In the real world, $P_2$ outputs their key $\pkey'$ and sends $c', \tau'$ to $\adv$.
        \item In the ideal world, $\Sim$ must send $(\NewKey, \sid, P_2, \cdot)$ to $\FPake$ and $c', \tau'$ to $\adv$. Let $\pwg$ be $\Sim$'s password guess for $(\sid, P_2)$.
    \end{itemize}
    \item Compute $\sk_i = \sk(r_i), K = \Dec(\sk_i,c')$, and $\pkey^* = \cK(K)$. 
    \item Let $\pkey'$ be the output of $P_2$ in either world. If $\pkey^* = \pkey'$ output $1$, otherwise output $0$.
\end{enumerate}
Our environment $\Env$ simply chooses $b' \gets \{0,1\}$ and runs $\Env_{b'}$.

\paragraph{Real world.} We show that $\Env$ outputs $1$ in the real world with $1-\mathsf{negl}$ probability.
\\
After $\Env_0$ queries $\ICEnc$ in Step 3 the state is
$$\frac{1}{\sqrt{2|\PK|}} \sum_{u \in \{0,1\}, \pk \in \PK} \ket{\pw_u}_{\Ar}\ket{\pw_u}_{\Xr}\ket{\pk}_{\Br}\ket{\pk}_{\Yr}\ket{\ICEnc(\pw_u, \pk)}_{\Zr}.$$ After Step 5, since $\ICEnc$ is a keyed permutation $\rho^0 = \ketbra{\psi}$, so $\Env_0$ must output $1$ with probability $1$.
\\
\\
After $\Env_1$ queries $\ICEnc$ in Step 3 the state is $$\frac{1}{\sqrt{2|\R|}} \sum_{u \in \{0,1\}, r \in \R} \ket{\pw_u}_{\Ar}\ket{\pw_u}_{\Xr}\ket{r}_{\Br}\ket{\pk(r)}_{\Yr}\ket{\ICEnc(\pw_u, \pk(r))}_{\Zr}.$$ If $\pw_i = \pw_b$ in Step 6 then $P_2$ will decrypt $\phi$ to obtain $\pk_i$ in Step 7. By correctness of $\KEM$, $P_2$ must then output $\pkey^*$ as computed by $\Env_1$ in Step 8 with probability $1- \mathsf{negl}$. Therefore $\Env_1$ outputs $1$ with probability $1 - \mathsf{negl}$.

\paragraph{Ideal world.} We now give a hybrid argument to show that in the ideal world, $\Env$ outputs $1$ with probability $\leq 1- u(\lambda)$ for some non-negligible $u(\lambda)$. Let $\Hyb_0$ be the experiment where $\Env$ runs in the ideal world against some simulator $\Sim$ for the dummy adversary. In general let $\Env_j^i$ be $\Env_j$ running in experiment $\Hyb_i$.
\\
\\
$\Hyb_1$: same as $\Hyb_0$, but at the end of Step 7 of $\Env_1$, output $1$ if $\pwg = \pw_i$; otherwise output $0$.
\\
\\
Analysis of the change: We can argue $ \Pr[\Hyb_1 = 1] \geq  \Pr[\Hyb_0 = 1] - \mathsf{negl}(\lambda)$ as follows. First, immediately outputting 1 when $\pwg = \pw_i$ only potentially increases the probability of outputting 1 when $\pwg = \pw_i$. Next, when $\pwg \neq \pw_i$, note in Step 7 $\pw_i = \pw_b$ so at this point $\Sim$ has not compromised $P_2$ with a correct password guess. Recall that in this case $\FPake$ samples $P_2$'s key $\pkey'$ uniformly and independently from $\{0,1\}^\lambda$, and it remains independent of $\Sim$'s view for the rest of the experiment. In total, we have $\Pr[\Hyb_1 = 1] \geq  \Pr[\Hyb_0 = 1] - 2^{-\lambda}$.
\\
\\
$\Hyb_2$: same as $\Hyb_1$, but change Step 2 of $\Env_1$ to be the same as $\Env_0$ and adjust the rest of $\Env_1$ accordingly.
\\
\\
Analysis of the change: In the previous hybrid, we made Steps 8 and 9 of $\Env_1$ unnecessary, which is the only place where $r_i$ is used. Therefore we no longer need the randomness on register $\Br$ and can use the same state as in $\Env_0$. We now have $\Env_1^2$ as follows:
\paragraph{Environment $\Env_1^2$}
\begin{enumerate}
    \item Sample $b \gets \{0,1\}$ and send $(\NewS, \sid, P_2, P_1, \pw_b,\mathsf{respondent})$ to $P_2$ for some nonce $\sid$.
    \item Prepare the normalized state \[\frac{1}{\sqrt{2|\PK|}} \sum_{u \in \{0,1\}, \pk \in \PK} \ket{\pw_u}_{\Ar}\ket{\pw_u}_{\Xr}\ket{\pk}_{\Br}\ket{\pk}_{\Yr}.\]
    \item Query $\ICEnc$ on registers $\Xr, \Yr$ into a new register $\Zr$ (through $\adv$).
    \item Measure $\Zr$ to get $\phi$ and leftover state $\rho^1_{\Ar \Br \Xr \Yr}$.
    \item Query $\ICDec(\pw_u, \phi) = \pk_u$ for $u \in \{0,1\}$.
    \item Measure $\Ar$ to get $\pw_i$; if $\pw_i \neq \pw_b$ output $1$.
    \item Instruct $\adv$ to send $\phi$ to $P_2$.
    \begin{itemize}
        \item In the real world, $P_2$ outputs their key $\pkey'$ and sends $c', \tau'$ to $\adv$.
        \item In the ideal world, $\Sim$ must send $(\NewKey, \sid, P_2, \cdot)$ to $\FPake$ and $c', \tau'$ to $\adv$. Let $\pwg$ be $\Sim$'s password guess for $(\sid, P_2)$.
    \end{itemize}
    \item If $\pwg = \pw_i$ output 1, otherwise output 0.
\end{enumerate}
The only functional change in this hybrid is that $\pk$ is now sampled uniformly, so we construct a quantum polynomial time reduction $\B$ that breaks computational public key uniformity given access to an $\Sim$ that distinguishes $\Hyb_1$ and $\Hyb_2$. Given the public key state $\eta_{\Yr}$, $\B$ runs $\Hyb_2$ with $\Sim$, playing the part of both $\FPake$ and the environment. $\FPake$ is simulated by running the functionality, and the environment is simulated by running $\Env^2$, with the following changes: 
\begin{itemize}
    \item If $b' = 0$ or $\pw_i \neq \pw_b$ then output a random bit.
    \item In Step 2 of $\Env_1$ instead prepare the state $$\frac{1}{\sqrt{2}} \sum_{u \in \{0,1\}} \ket{\pw_u}_{\Ar}\ket{\pw_u}_{\Xr}\eta_{\Yr}.$$
\end{itemize}
Note that \begin{itemize}
    \item $\B$ must use its access to $\FPake$ in Step 8 of $\Env_1$ to check if $\pwg = \pw_i$;
    \item Changing Step 2 of $\Env_1$ does not affect its execution since $\Env_1$ no longer uses the $\Br$ register.
\end{itemize}
Let $c$ be the choice bit in the uniformity experiment. When $c = 0$ the view of $\Sim$ in $\B$ is as in $\Hyb_2$, and when $c = 1$ the view of $\Sim$ in $\B$ is as in $\Hyb_1$. Note that the events $b' = 0$ and $\pw_i = \pw_b$ happen independently in $\Hyb_1, \Hyb_2$ and the execution of $\B$ (regardless of the value of $c$). Furthermore, when $b' = 0 \lor \pw_i \neq \pw_b$ then both $\Env^1$ and $\Env^2$ behave identically. Similarly, in this case $\B(c=0)$ and $\B(c=1)$ behave identically. Concluding,
\begin{align*}
    &|\Pr[\Env^2 \text{ outputs } 1] - \Pr[\Env^1 \text{ outputs } 1]| \\ & = |\Pr[\B(c=0) \text{ outputs } 1] - \Pr[\B(c=1) \text{ outputs } 1]| \\ &= \mathsf{Adv}(\text{PKU}_\KEM^\B) \leq \mathsf{negl}(\lambda).
\end{align*}
$\Hyb_3$: Let $\Wr$ the additional state in the experiment besides $\Ar \Br \Xr \Yr$: this includes $\Sim$'s workspace register, the state of $\FPake$ and other classical public values. Let $\delta^j$ be the state over $\Ar \Br \Xr \Yr \Wr$ at the end of Step 5 of $\Env_j^2$. $\Hyb_3$ is the same as $\Hyb_2$ but $\delta^j$ is replaced with  $\frac{(\Pi_\psi \otimes I_W)\delta^j(\Pi_\psi \otimes I_W)}{1-p}$ at the end of Step 5 of $\Env_j^2$ for $j \in \{0,1\}$ (here $p = \Pr[M \text{ rejects in Step 6 of } \Env_0]$).
\\
\\
Analysis of the change: Since $\Env_0^2, \Env_1^2$ are the same through the end of Step 5, $\delta^0 = \delta^1 = \delta$. Since $\rho^j = \Tr_{\Wr}(\delta^j)$, $\rho^0 = \rho^1 = \rho$. If $p$ is non-negligible then $\Env_0$ outputs 1 with probability $\leq 1- p$ and we are done, so we may assume that $p$ is negligible. By the Gentle Measurement Lemma, $$|\Pr[\Env^3 \text{ outputs } 1] - \Pr[\Env^2 \text{ outputs } 1]| \leq O(\sqrt{p}).$$
\\
We now show $\Pr[\Env^3 \text{ outputs 1}] \leq 7/8$. This completes the proof since by our hybrid argument 
$$\Pr[\Env^0 \text{ outputs 1}] \leq \Pr[\Env^3 \text{ outputs 1}] + \mathsf{negl} \leq 7/8 + \mathsf{negl}.$$
Let $\gamma$ be the state over $\Ar \Br \Xr \Yr \Wr$ at the end of Step 5 of $\Env_j^3$ (by previous reasoning $\gamma$ does not depend on $j$). By the change in $\Hyb_3$, there is a state $\sigma_{\Wr}$ such that \begin{align}
    \gamma = \ketbra{\psi}_{\Ar \Br \Xr \Yr} \otimes \sigma_{\Wr}.\label{pwindep}
\end{align}
We claim that the instant before $\Sim$ sends $\TestPwd$ (if it does), $b$ is uniform from it's point of view. If $\TestPwd$ is sent before the end of Step 5, this holds since $\Sim$'s view is independent of $b$ at this point. If it is sent after, by \zcref[S]{pwindep} the value of $i$ is independent of $\Sim$'s state, and Step 7 reveals no information about $i$ before $\TestPwd$ is sent. In Step 7 we have $\pw_b = \pw_i$, so either way $\Pr[\pwg = \pw_i] = 1/2$. Since $i,b$ are sampled independently and $b$ is uniform, $\Env^3$ reaches Step 7 of $\Env_1^3$ with probability $1/4$, so it must output $0$ with probability $\geq 1/8$.
\end{proof}

\section{Provable game-based security}

\subsection{Masny-Rindal OT}\label{subsec:MR-game-based}

\begin{theorem}
    Let $\Pi = (\Rec_1,\Sen,\Rec_2)$ be the Masny-Rindal OT protocol (\zcref[S]{subsubsec:MR-syntax}) instantiated with any KEM with negligible correctness error (\zcref[S]{def:KEM-correctness}), (post-quantum) computationally uniform public keys (\zcref[S]{def:KEM-unikeys}), and (post-quantum) key unpredictability (\zcref[S]{def:KEM-unp}). Then $\Pi$ satisfies correctness (\zcref[S]{def:OT-correctness}), (post-quantum) game-based security against malicious receiver (\zcref[S]{def:OT-sender}), and (post-quantum) game-based security against malicious sender (\zcref[S]{def:OT-receiver}).\label{thm:MRgbsec}
\end{theorem}
\noindent Correctness follows immediately from correctness of $\KEM$.
\begin{proof}[Proof of malicious receiver security]
     As an intermediate step, we first prove a weaker version where the malicious receiver has to guess both strings. A QPT adversary $(\cA_1, \cA_2)$ should have negligible success in the following experiment:\footnote{We assume that the sender aborts if $t_0 = t_1$, which only occurs with $\mathsf{negl}(\lambda)$ probability in an honest execution.}
    \begin{itemize}
        \item On $1^\lambda$, $\cA^H_1(1^\lambda)$ runs and outputs a message $(t_0, t_1)$ and auxiliary state $\sigma$.
        \item For $j \in \{0,1\}$ sample random coins $r_j \gets \R$ and compute $\pk_j = t_j \oplus H(t_{1-j}),\allowbreak (c_j, K_j) \coloneqq \Enc(\pk_j;r_j)$.
        \item Run $\cA_2^H(1^\lambda,c_0,c_1,\sigma)$ to get $(K_0^*, K_1^*)$.
        \item Adversary succeeds if $K_j = K_j^*$ for $j \in \{0,1\}$.
    \end{itemize}
    We now define a reduction $\B^{\cA_1,\cA_2}$ that uses measure-and-reprogram to break key unpredictability. $\B$ will reprogram $H$ such that the public key $\pk_{\pi(0)}$ equals the challenge public key $\pk^*$. $\B$ then uses the challenge ciphertext $c^*$ in place of $c_{\pi(0)}$ to complete the protocol, and uses $\cA_2$'s guess for $K^*_{\pi(0)}$ as its own.
    \\
    \\
    Formally, suppose $\cA_1$ makes $q_1$ queries to $H$ and let $\delta$ be the success probability of $(\cA_1, \cA_2)$ in the weaker malicious receiver security game described above. Let $C^H(t_0, t_1, H(t_0), H(t_1), \sigma)$ be the following quantum algorithm:
    \begin{itemize}
        \item Compute $\pk_j \coloneqq t_j \oplus H(t_{1-j})$ for $j \in \{0,1\}$;
        \item Sample $r_0, r_1 \gets \R$;
        \item Compute $(c_j, \ast) \coloneqq \Enc(\pk_j;r_j)$;
        \item Write $K_0^*, K_1^* \gets \cA_2^H(c_0, c_1, \sigma)$ to an output register.
    \end{itemize}
    Let $\Pi_{\mathsf{win}}$ be the projector onto states where $K_j^* = \Enc(\pk_j;r_j)[1]$ for $j \in \{0,1\}$. Let $V^H(t_0,\allowbreak t_1,\allowbreak H(t_0),\allowbreak H(t_1),\allowbreak \sigma)$ output $1$ when the projector $(C^H)^\dagger \Pi_{\mathsf{win}}C^H$ accepts.
    \\
    $\B$ is given an unpredictability challenge $\pk^*,c^*$ and will act as the measure-and-reprogram simulator for $\cA_1^H$, with predicate $V^H$. In the first stage $\B$ receives $\pi, t_{\pi(0)}$ and sets $\Theta_{\pi(0)} \gets \PK$; in the second stage $\B$ receives $t_{\pi(1)}$ and sets $\Theta_{\pi(1)} = \pk^* \oplus t_{\pi(0)}$. On output $t_0, t_1, \sigma$, $\B$ runs $C^{H(\mathbf{t} \ast \mathbf{\Theta})}$ except that it does not generate $c_{\pi(0)}$ itself but instead uses $c^*$. $\B$ then measures and outputs $K_{\pi(0)}^*$. To analyze the success of $\B$, first note that $V^H$ accepts on $\cA_1$'s output exactly when $(\cA_1, \cA_2)$ wins the weaker malicious receiver security game. Additionally, $\B$ sets $\pk^* = \pk_{\pi(0)} = t_{\pi(0)} \oplus H(t_{\pi(1)}), c^* = c_{\pi(0)}$, so $K_{\pi(0)}^*$ will equal the desired challenge key if $(\cA_1, \cA_2)$ succeeds. The only difference in $(\cA_1, \cA_2)$'s view between $\B$ and the usual measure-and-reprogram simulator is that $\Theta_{\pi(1)}$ is not uniformly random. However, this cannot affect $\B$'s success probability too much without breaking public key uniformity: a reduction $\B'$, given a uniformity challenge $\pk^*$, encapsulates $(c^*, K^*) \gets \Enc(\pk^*)$, and runs $\B$ as above, outputting $0$ if $K_{\pi(0)}^* = K^*$ (and $1$ otherwise). All together, summing over all choices of $\mathbf{t}^*$ and applying \zcref[S]{thm:mandp} with $t=2$ gives \begin{align*}
        \mathsf{Adv}(\text{UNP}_{\KEM, \lambda}^{\B}) \geq \frac{\delta}{(2q_1+1)^4} - \mathsf{Adv}(\text{PKU}_{\KEM, \lambda}^{\B'}),
    \end{align*} completing the proof. We obtain full security against a malicious receiver by using \zcref[S]{lemma:binrewind} to reduce to the weaker version. Suppose $(\cA_1, \cA_2)$ is a QPT adversary for malicious receiver security of Masny-Rindal OT with success probability $1/2 + \epsilon$:
    \begin{itemize}
        \item On $1^\lambda$, $\cA^H_1(1^\lambda)$ runs and outputs a message $(t_0, t_1)$ and auxiliary state $\sigma$.
        \item For $j \in \{0,1\}$ sample random coins $r_j \gets \R$ and compute $\pk_j = t_j \oplus H(t_{1-j}),\allowbreak (c_j, K_j) \coloneqq \Enc(\pk_j;r_j)$. Sample $b \gets \{0,1\}$.
        \item Run $\cA_2^H(1^\lambda,c_0,c_1,b,\sigma)$ to get $K_b^*$.
        \item Adversary succeeds if $K_b = K_b^*$.
    \end{itemize}
    Without loss of generality, on input $1^\lambda,c_0,c_1,b,\sigma$, $\cA_2^H$ runs an oracle-aided unitary $U_b^H$ (with $1^\lambda,c_0, c_1$ hardwired in) on $\sigma$ (plus some ancilla qubits), writing $K_b^*$ into an extra register $\Xr$. Let $P^*_b$ project onto states with \[\Xr = \Enc(\pk_b;r_b)[1]\] and define $P_b = (U_b^H)^\dagger P^*_b U_b^H$. We now define an adversary $(\cA'_1, \cA'_2)$ for the weaker security:
    \begin{itemize}
        \item $\cA'_1$: Run $\cA_1^H$ as normal, outputting $t_0,t_1,\sigma$.
        \item $\cA'_2(c_0,c_1,\sigma)$: \begin{enumerate}
            \item Apply $U_0^H$ to $\sigma$, measure $\Xr$ to obtain $K_0^*$, and copy $K_0^*$ to a fresh register.
            \item Apply $U_1^H (U_0^H)^\dagger$ and measure $\Xr$ to get $K_1^*$. Output $(K_0^*,K_1^*)$.
        \end{enumerate}
    \end{itemize}
    For any fixed choice of $r_0, r_1$, applying \zcref[S]{lemma:binrewind} to the projectors $P_i$ and (in expectation over) the state $\sigma$, we have $\Tr(P_0\sigma) + \Tr(P_1\sigma) = 1 + 2\epsilon$ by assumption and the unitarity of $U_i^H$, so $\Tr(P_0P_1P_0\sigma) \geq \epsilon^2$. Since $P_i$ projects $\Xr$ to a single value, the measurement of $\Xr$ does not disturb the state conditioned on success, and the probability that $K^*_j = K_j$ for $j \in \{0,1\}$ is exactly $\Tr(P_1P_0\sigma P_0P_1) = \Tr(P_0P_1P_0\sigma) \geq \epsilon^2$. By averaging over the sampling of $r_0, r_1$, $(\cA_1,\cA_2)$ must have negligible advantage in the malicious receiver security game.
\end{proof}
\begin{proof}[Proof of malicious sender security] The malicious sender security game for Masny-Rindal OT against QPT adversary $\cA$ is the following:
    \begin{itemize}
        \item On $1^\lambda$, sample $b \gets \{0,1\}$, $t_{1-b} \gets \PK$, and $(\pk_b,\sk_b) \gets \KG(1^\lambda)$.
        \item Set $t_b \coloneqq \pk_b \oplus H(t_{1-b})$.
        \item Send $\msg_R \coloneqq (t_0,t_1)$ to $\cA$ and keep $\sk \coloneqq \sk_b$.
        \item $\cA$ returns a guess $b^*$: $\cA$ wins if $b^* = b$.
    \end{itemize}
    Suppose $\cA$ is an adversary for malicious sender security with success probability $\frac{1}{2} + \epsilon, \epsilon \geq 0$. We give a reduction $\B^{\cA}$ that breaks uniformity of public keys. Given a challenge $\pk^*$ with choice bit $c$, $\B$ plays the part of $\Rec_1$ but uses $\pk^*$ instead of generating its own $\pk_b$. If $b^* = b$ then $\B$ outputs $1$; otherwise it outputs a random bit. Note that when $c = 0$ ($\pk^* \gets \PK$) the bit $b$ is independent of $t_0, t_1$ so \[\Pr[b^* = b | c = 0] = 1/2.\] When $c = 1$ ($\pk^*$ honestly generated) $\B$ simulates the malicious sender security experiment to $\cA$, so $\Pr[b^* = b | c = 1] = 1/2 + \epsilon$, which implies $\Pr[b^* = b] = 1/2 + \epsilon/2$. We then have
    \begin{align*}
        \Pr[\mathsf{win}] = \Pr[\mathsf{win} | b^* = b]\left(\frac{1}{2} + \frac{\epsilon}{2}\right) + \Pr[\mathsf{win} | b^* \neq b]\left(\frac{1}{2} - \frac{\epsilon}{2}\right) \\
        = \Pr[\mathsf{win} | b^* = b]\left(\frac{1}{2} + \frac{\epsilon}{2}\right) + \frac{1}{2}\left(\frac{1}{2} - \frac{\epsilon}{2}\right).
    \end{align*}
    By Bayes Theorem,
    \[
    \Pr[\mathsf{win} | b^* = b] = \Pr[c = 1 | b^* = b] = \frac{\frac{1}{2}\left(\frac{1}{2} + \epsilon\right)}{\left(\frac{1}{2} + \frac{\epsilon}{2}\right)}.
    \]
    All together,\[
    \Pr[\mathsf{win}] = \frac{1}{2}\left(\frac{1}{2} + \epsilon\right) + \frac{1}{2}\left(\frac{1}{2} - \frac{\epsilon}{2}\right) = \frac{1}{2} + \frac{\epsilon}{4}.
    \]
\end{proof}

\subsection{OEKE with 2-Feistel}

\begin{theorem}
    Let $\Pi$ be any OEKE protocol instantiated with (1) any KEM with negligible correctness error, UNP security, and computationally uniform public keys (see \zcref[S]{subsec:KEM}) (2) the 2-Feistel-based encryption and decryption algorithms $\cE_\TF,\cD_\TF$ (\zcref[S]{subsec:attack-general}), (3) key derivation function $\cK$ modeled as a \emph{salted} random oracle $\Hkey$ with $\lambda$-length output,\footnote{Salting the random oracle means that the respondent will sample a random salt $w \gets \{0,1\}^\lambda$, set $\pkey = \Hkey(w,K)$ as the output of the key derivation function, and send $w$ as part of its message, meaning that the key derivation function is effectively $\Hkey(w,\cdot)$.} and (4) tag function $\cT$ modeled as a random oracle $\Htag$ with $\lambda$-length output. 
    
    Then $\Pi$ satisfies Correctness (\zcref[S]{def:PAKE-correctness}), Passive security (\zcref[S]{def:PAKE-passive}), Security against malicious initiator (\zcref[S]{def:PAKE-initiator}), and Security against malicious respondent (\zcref[S]{def:PAKE-respondent}).
\end{theorem}
\noindent Correctness is immediate. Passive security follows from the security of the KEM; although the adversary can decrypt $\phi$ using $\pw$ to get $\pk$, it still only sees $\pk,c,\tau,w$, and has to distinguish between $\Hkey(w,K)$ and a uniformly random string. Key unpredictability of KEM plus standard one-way-to-hiding shows that the adversary has negligible advantage in this game.
\begin{proof}[Proof of security against a malicious initiator]
We first establish a weaker form of security that we call the \emph{two-output} game: the malicious initiator should not be able to guess two independently encapsulated keys. Formally, a QPT adversary $(\cA_1, \cA_2)$ should have negligible probability of success in the following experiment:
\begin{itemize}
    \item Sample $\pw_1 \gets D$, followed by $\pw_2 \gets D$ conditioned on $\pw_2 \neq \pw_1$.
    \item $(\phi^*,\sigma) \gets \cA_1^{H_1,H_2}(1^\lambda)$ outputs $\phi^*$ and auxiliary state $\sigma$.
    \item Sample $r_1, r_2 \gets \R$.
    \item Compute $\pk_i = \cD_{\TF}(\pw_i,\phi^*)$ and $(c_i,K_i) \coloneqq \Enc(\pk_i;r_i)$ for $i \in \{1,2\}$.
    \item Run $(K_1^*, K_2^*) \gets \cA_2^{H_1,H_2}(\pw_1,\pw_2,c_1,c_2,\sigma)$; the adversary succeeds if $K_i^* = K_i$ for $i \in \{1,2\}$.
\end{itemize}

To show this, we will require a variant of \zcref[S]{thm:mandp} that works for multiple random oracles. We model this as follows: suppose we have random oracles $H_1: A_1 \to B_1, H_2: A_2 \to B_2$ and $\cA_1^{H_1, H_2}$ is a quantum oracle machine given access to the $H_i$ through the unitaries $U_i: \ket{x,y} \mapsto \ket{x,H_i(x) \oplus y}$. This adversary can be emulated by an $\cA_2$ making the same number of queries given access to the random oracle $H: (\{1\} \times A_1) \sqcup (\{2\} \times A_2) \to B_1 \times B_2$ by querying $\ket{i,a_i}$ and only using the relevant part of the output. We denote a query $H_i(x)$ as $H(i,x)$ or $(i,x)$. Let $H(i,x)$ denote the intended random oracle output $b_i$; we also occasionally abuse notation and denote the full output $(b_1, b_2) \in B_1 \times B_2$ by $H(i,x)$ as well. We may write an oracle algorithm $\cA^{H_1, H_2}$ as $\cA^H$ under this identification.
\\
\\
Suppose that the two-output game has non-negligible success probability. By averaging, fix distinct passwords $\pw_1,\pw_2$ for which the conditional success probability, denoted $\delta$, is non-negligible. Define an algorithm $\cA_1'$ as follows:
\begin{itemize}
    \item Run $(\phi^* = (s,T), \sigma) \gets \cA_1^H(1^\lambda)$;
    \item Output \begin{align*}
        A_1 = (2,\pw_1,T),A_2 = (2,\pw_2,T),\phi^*, \sigma, \\ B_1 = (1,\pw_1,u_1 = s \oplus H(2,\pw_1,T)),B_2 = (1,\pw_2,u_2 = s \oplus H(2,\pw_2,T)).
    \end{align*}
\end{itemize}
\noindent Suppose $\cA_1'$ makes $q_1$ queries to $H$. Let $C[r_1,r_2]^H$ be a quantum algorithm that has hard-coded random coins $r_1,r_2$, additionally takes inputs 
\begin{align}
        A_1 = (2,\pw_1,T),A_2 = (2,\pw_2,T),\phi^* = (s,T), \sigma,\label{eq:input-form1} \\ B_1 = (1,\pw_1,u_1),B_2 = (1,\pw_2,u_2),\Theta_{A_1}^{(2)},\Theta_{A_2}^{(2)},\Theta_{B_1}^{(1)},\Theta_{B_2}^{(1)},\label{eq:input-form2}
    \end{align}
and does the following:
\begin{itemize}
    \item Write $1$ to an ancilla $\flag$ unless $u_i = s \oplus \Theta_{A_i}^{(2)}, H(A_i) = \Theta_{A_i}^{(2)}, H(B_i) = \Theta_{B_i}^{(1)}$ for all $i \in \{1,2\}$.
    \item Continue execution if $\flag = 0$: 
    \item Compute $\pk_i = (\Theta_{B_i}^{(1)})^{-1} \odot T$.
    \item Compute $(c_i,K_i) \coloneqq \Enc(\pk_i;r_i)$ for $i \in \{1,2\}$.
    \item Run $(K_1^*, K_2^*) \gets \cA_2^{H_1,H_2}(\pw_1,\pw_2,c_1,c_2,\sigma)$ to an output register.
    \item Write $1$ to $\flag$ unless $K_j^* = \Enc(\pk_j;r_j)[1]$ for $j \in \{1,2\}$.
\end{itemize}
Let $\Pi_{\mathsf{win}}$ be the projector onto output states where the value of $\flag$ equals $0$. Using a purified implementation of $C[r_1,r_2]^H$, let $V[r_1,r_2]^H$ output $1$ when the projector $(C[r_1,r_2]^H)^\dagger \Pi_{\mathsf{win}}C[r_1,r_2]^H$ accepts, and $0$ otherwise. We will show that \[\Pr_{r_1,r_2 \gets \R}\left[V[r_1,r_2]^H(A_1,A_2,\phi^*,\sigma,B_1,B_2) = 1 : (A_1,A_2,\phi^*,\sigma,B_1,B_2) \gets \cA_1'\right] = \mathsf{negl}(\lambda).\]

Fix some choice of $A_1, A_2, B_1, B_2$. Recall that the measure-and-reprogram simulator for $\cA_1'$ will measure the points $A_1, A_2, B_1, B_2$ to reprogram in some order $\prec$ based on a random permutation $\pi$.
\begin{claim}\label{2fprog}
    Let $E_1$ be the event that $B_1 \prec A_1, B_2 \prec A_2$ and $E_2$ be the event that $V^H$ outputs $1$. Then $\Pr[E_1 \land E_2] \leq 2^{-\ell}$.
\end{claim}
\begin{proof}
    Suppose without loss of generality that $A_1 \prec A_2$. At the point when $H(A_2)$ is reprogrammed to $\Theta_{A_2}^{(2)}$, if $E_1$ occurs then the values of $u_1, u_2,\Theta_{A_1}^{(2)}$ are known to the simulator, since the $B_1, B_2,A_1$ queries have already been made. If $E_2$ occurs later on we must have $u_2 \oplus \Theta_{A_2}^{(2)} = u_1 \oplus \Theta_{A_1}^{(2)}$, since these values will equal $s$ when $s$ is defined later on. Since $\Theta_{A_2}^{(2)}$ is a uniformly random value in $\{0,1\}^\ell$ the claim follows.
\end{proof}
\zcref[S]{2fprog} will allow us to program a public key $\pk_{i^*}$. Formally, we give a reduction $\B$ to key unpredictability of $\KEM$ as follows. Given a challenge $\pk^*, c^*$
\begin{itemize}
    \item $\B$ will act as the measure-and-reprogram simulator for $\cA_1'$.
    \item $\B$ reprograms queries to random outputs $\Theta_i = \Theta_i^{(1)} \lVert \Theta_i^{(2)}$ until it measures $(1,\pw_{i^*}, u_{i^*})$ after previously measuring $(2,\pw_{i^*}, T)$ for some $i^*,T$. In this case, it reprograms the $H_1$ part of the query to be $(\pk^*)^{-1} \odot T$, and the $H_2$ part to be a random $\Theta_{i^*}^{(2)}$. Otherwise it aborts the measure-and-reprogram and just guesses a random key $K \gets \{0,1\}^\lambda$.
    \item On outputs $\mathbf{x} = (A_1, A_2, B_1, B_2)$, $\B$ samples $r_{3-i^*} \gets \R$ and runs $C[\{r_{3-i^*},r_{i^*}\}]^{H(\mathbf{x} \ast \mathbf{\Theta})}$ except that it uses $c^*$ in place of $c_{i^*}$ (and hence does not actually need to know the value of $r_{i^*}$, which was sampled by its challenger). $\B$ then measures and outputs $K_{i^*}^*$.
\end{itemize}
Note that $V[\{r_{3-i^*},r_{i^*}\}]^H$ outputs $1$ exactly when $(\cA_1', \cA_2)$ wins the security game (with probability $\delta$). We now argue that $V[\{r_{3-i^*},r_{i^*}\}]^H$ still outputs $1$ in $\B$ with noticeable probability. Assuming that a suitable $i^*$ always exists incurs a negligible loss; by \zcref[S]{2fprog}, \[ \Pr[V[\{r_{3-i^*},r_{i^*}\}]^H \text{ outputs } 1 \land  \exists \ \text{such an } i^*] \geq \Pr[V[\{r_{3-i^*},r_{i^*}\}]^H \text{ outputs } 1] - 2^{-\ell}. \]
The reprogramming of $H_1$ to $(\pk^*)^{-1} \odot T$ is indistinguishable to $(\cA_1', \cA_2)$ from reprogramming to a uniform value: a reduction $\B'$ to public key uniformity can, on receiving a challenge $\pk^*$, compute an encapsulation $(c^*,K^*) \gets \Enc(\pk^*)$, and run $\B$ using $\pk^*,c^*$, returning $0$ if $K^* = K^*_{i^*}$ and $1$ otherwise. If $\pk^*$ is random $\B'$ simulates the usual measure-and-reprogram simulation, and if $\pk^*$ is real $\B'$ simulates the reduction $\B$. Therefore, summing over all points $A_1, A_2, B_1, B_2$ of the form specified in \ref{eq:input-form1} \& \ref{eq:input-form2} and applying \zcref[S]{thm:mandp} when $t = 4$ we have \begin{align*}
        \mathsf{Adv}(\text{UNP}_{\KEM, \lambda}^{\B}) \geq \frac{\delta}{(2q_1+1)^8} - \mathsf{Adv}(\text{PKU}_{\KEM, \lambda}^{\B'}) - 2^{-\ell},
    \end{align*} completing the proof. In fact, the above argument can be extended to the case where $\cA_2$ is additionally given $w_i,y_i = \Hkey(w_i,K_i), \tau_i = \Htag(\pw_i,\pk_i,\phi^*,c_i,K_i)$ for $i \in \{1,2\}$, and $(\cA_1,\cA_2)$ have quantum access to $\Hkey, \Htag$. To see this, we apply standard one-way-to-hiding \cite{EC:Unruh14} within the KEM reduction, replacing the embedded challenge’s tag and derived-key value by uniform strings while computing the other response honestly. 

    
    Using this, we now show a stronger version of search security: any QPT adversary $(\cA_1, \cA_2)$ has at most $2^{-\gamma} + \mathsf{negl}$ chance of success in the following experiment:
    \begin{itemize}
    \item $(\phi^*,\sigma) \gets \cA_1^{H_1,H_2,\Hkey,\Htag}(1^\lambda)$ outputs $\phi^*$ and auxiliary state $\sigma$.
    \item Sample $\pw \gets D, r \gets \R$, $w \gets \{0,1\}^\lambda$.
    \item Compute $\pk = \cD_{\TF}(\pw,\phi^*)$, $(c,K) \coloneqq \Enc(\pk;r)$, $y = \Hkey(w,K)$ and $\tau = \Htag(\pw,\pk,\phi^*,c,K)$.
    \item Run $K^* \gets \cA_2^{H_1,H_2,\Hkey,\Htag}(\pw,c,\tau,w,y,\sigma)$; the adversary succeeds if $K^* = K$.
\end{itemize}

We use \zcref[S]{multirewind} to reduce to the weaker version shown previously. Sample and fix a choice of $r_\pw \gets \R$ and $w_\pw \gets \{0,1\}^\lambda$ for each password $\pw$. We will define a family of unitaries $\{U_\pw\}_\pw$, where each $U_\pw$ has a tuple $(c,\tau,w,y)$ hard-coded, which are deterministically derived from $\pw,r_\pw,w_\pw$, the random oracles, and $\phi^*$. $U_\pw$ applies a coherent implementation of $\cA_2^{H_1,H_2,\Hkey,\Htag}$ to $\pw,c,\tau,w,y,\sigma$ (plus some ancilla qubits), writing $K^*$ into an extra register $\Xr$. Let $P^*_{\pw}$ project onto states with \[\Xr = \Enc(\pk_\pw;r_\pw)[1],\] where $\pk_{\pw} = \cD_{\TF}(\pw,\phi^*)$, and define \[P_{\pw} = (U_{\pw}^H)^\dagger P^*_{\pw} U_{\pw}^H.\] Given $(\cA_1, \cA_2)$ with success probability $2^{-\gamma} + \epsilon$ we now define an adversary $(\cA'_1, \cA'_2)$, where $\pw_1,\pw_2$ are sampled as in the two-output game:
    \begin{itemize}
        \item $\cA'_1(1^\lambda)$: output $(\phi^*,\sigma) \gets \cA_1^{H_1,H_2,\Hkey,\Htag}(1^\lambda)$.
        \item $\cA'_2(\pw_1,\pw_2,\sigma)$: \begin{enumerate}
            \item Apply $U_{\pw_1}^H$ to $\sigma$, measure $\Xr$ to obtain $K_1^*$, and copy $K_1^*$ to a fresh register.
            \item Apply $U_{\pw_2}^H (U_{\pw_1}^H)^\dagger$ and measure $\Xr$ to get $K_2^*$.
            \item Output $(K_1^*,K_2^*)$.
        \end{enumerate}
    \end{itemize}
    As in \zcref[S]{thm:MRgbsec}, since each $P_\pw^*$ projects $\Xr$ to a single value, the measurement of $\Xr$ does not disturb the state conditioned on success. By applying \zcref[S]{multirewind} to the projectors $\{P_\pw\}_\pw$ and the state $\sigma$, averaged over $H_1,H_2,\Hkey,\Htag$,$\phi^*,\sigma$,$\{r_\pw,w_\pw\}_\pw$, we get that the probability $(\cA'_1, \cA'_2)$ outputs $K_1^* = K_1, K_2^* = K_2$ is at least $\epsilon^3$. This implies the existence of an adversary $(\cA'_1, \cA'_2)$ with the same probability of success in the two-output game above, where $\cA'_2$ now additionally takes $(c_1,\tau_1,w_1,y_1)$, $(c_2,\tau_2,w_2,y_2)$ as input and hard-codes them into the descriptions of $U_{\pw_1},U_{\pw_2}$. But above we showed that this probability must be negligible, establishing that $\epsilon = \mathsf{negl}$.

    
    Finally we are ready to show full indistinguishability security against a malicious initiator. Here is where we use the assumption that $\Hkey$ is salted: a standard re-programming argument allows us to move to a hybrid (with $\mathsf{negl}$ distinguishing advantage) where the oracle $\Hkey(w,*)$ is sampled \emph{independently} of $K$ and the adversary's state after it sends its first message. We can apply \zcref[S]{cor:search-to-decision} with $k = K, H = \Hkey(w, \ast)$ and $\rho$ being the transcript/residual state. Therefore $|\Pr[\cA_2 \text{ outputs } 1 | b = 0] - \Pr[\cA_2 \text{ outputs } 1 | b = 1]| \leq 2^{-\gamma} + \mathsf{negl}$, and hence the probability of winning in the full game is at most $1/2 + (2^{-\gamma} + \mathsf{negl})/2 = 1/2 + 2^{-\gamma}/2 + \mathsf{negl}$.
\end{proof}

\begin{proof}[Proof of security against a malicious respondent]
    As an intermediate step, we first prove a weaker version of security where the malicious receiver has to guess $\pw$ given $\phi$. A QPT adversary $\cA$ should have at most $2^{-\gamma} + \mathsf{negl}$ probability of success in the following experiment, where $\pw \gets D, (*,\phi) \gets \PAKE_1(\pw); \pw^* \gets \cA^{H_1,H_2}(\phi);$ and  $\cA$ wins if $\pw = \pw^*$. Given $\cA$ that succeeds with $2^{-\gamma} + \epsilon$, we give a reduction $\B^{\cA}$ that breaks uniformity of public keys with advantage at least $\frac{\epsilon}{4}$. Given a challenge $\pk^*$ with choice bit $c$, $\B$ plays the part of $\PAKE_1$ but uses $\pk^*$ instead of honestly generating its own $\pk$. If $\pw^* = \pw$ then $\B$ outputs $1$; otherwise it outputs a uniformly random bit.
        \begin{claim}
            When $c = 0$ $(\pk^* \gets \PK)$,  $(\pw,\phi)$ are distributed as independent samples $\pw \gets D, \phi \gets \{0,1\}^\ell \times \PK$. \label{claim:phindep}
        \end{claim}
        \begin{proof}
            Let us consider fixed $H_1, H_2, \pw$. Recall that $\cE_\TF(\pw, \ast ; \ast)$ is a bijective map $\PK \times \{0,1\}^\ell \to \{0,1\}^\ell \times \PK$, with left inverse $\cD_\TF(\pw, \ast)$. Therefore if $w,\pk$ are chosen uniformly, then $\phi = \cE_\TF(\pw,\pk; w)$ will be uniform as well. Since $\pw$ is fixed here, when $\pw \gets D$ it must be independent of $\phi$. Additionally, since $H_1,H_2$ are fixed, this continues to hold even against an adversary that sees the entire truth tables of $H_1, H_2$ (so certainly against an adversary with polynomially many quantum queries to $H_1,H_2$).
        \end{proof}
        \zcref[S]{claim:phindep} immediately implies $p_0 \coloneqq \Pr[\pw^* = \pw | c = 0] \leq 2^{-\gamma}$. When $c = 1$, $\B$ simulates an honest initiator to $\cA^{H_1,H_2}$, so $p_1 \coloneqq \Pr[\pw^* = \pw | c = 1] = 2^{-\gamma} + \epsilon$; all together, $\Pr[\B \text{ guesses } c] = \frac{1}{2} + \frac{p_1 - p_0}{4} \geq \frac{1}{2} + \frac{\epsilon}{4}$. We complete the proof of security against a malicious respondent by using \zcref[S]{thm:simextract} to build a password-guessing adversary from a malicious respondent adversary. Suppose $(\cA_1,\cA_2)$ wins the experiment of \zcref[S]{def:PAKE-respondent} with probability $\frac{1}{2} + \frac{2^{-\gamma} + \epsilon}{2}$ for some non-negligible $\epsilon$. Let $\mathsf{Acc}$ denote the event that the honest initiator accepts $\cA_1$'s response. Since $\Pr[\mathsf{win}] \leq \frac{1}{2} + \frac{1}{2}\Pr[\mathsf{Acc}]$, we have $\Pr[\mathsf{Acc}] \geq 2^{-\gamma} + \epsilon$. We now construct a password-guessing adversary that succeeds with probability at least $2^{-\gamma} + \epsilon - \mathsf{negl}(\lambda)$, contradicting the above bound.

        Suppose $\cA_1$ makes $q_1$ queries to $\Htag$. Consider the following experiment $\Hyb_0$, where we run the honest initiator against $\cA_1$ and output all values:
        \begin{itemize}
            \item $\pw \gets D$
            \item $(\sk,\pk,\phi) \gets \PAKE_1(\pw)$;
            \item $c,\tau \gets \cA_1^{H_1,H_2,\Htag}(\phi)$;
            \item Compute $K = \Dec(\sk,c), \tau' = \Htag(\pw,\pk,\phi,c,K)$.
            \item Output $\pw,\pk,\phi,c,K,\tau, \tau'$.
        \end{itemize}
        Note that $\Pr[\tau' = \tau] \geq \Pr[\mathsf{Acc}] \geq 2^{-\gamma} + \epsilon$. 
        \\
        $\Hyb_1$: In this hybrid, we use the extractable QROM simulator of \zcref[S]{thm:simextract} to simulate $\Htag$. Since there are no $\Sim.E$ queries, $\Hyb_1$ and $\Hyb_0$ are perfectly indistinguishable.
        \\
        $\Hyb_2$: In this hybrid, we add a $\Sim.E$ query at the end to extract $\cA_1$'s password guess.
        \begin{itemize}
            \item $\pw \gets D$
            \item $(\sk,\pk,\phi) \gets \PAKE_1(\pw)$;
            \item $c,\tau \gets \cA_1^{H_1,H_2,\Sim.RO}(\phi)$;
            \item Compute
            \begin{align*}
                K &= \Dec(\sk,c), \\
                \tau' &= \Sim.RO(\pw,\pk,\phi,c,K), \\
                (\pw',\ast,\ast,\ast,\ast) &= \Sim.E(\tau).
            \end{align*}
            \item Output $\pw,\pk,\phi,c,K,\tau, \tau',\pw'$.
        \end{itemize}
        Since $\Sim.E$ is applied at the end, the other values are distributed identically between $\Hyb_2$ and $\Hyb_1$. By \zcref[S]{thm:simextract}, $\Pr[\tau = \tau' \land \pw \neq \pw'] \leq O(q_1^3/2^\lambda).$ Therefore in $\Hyb_2$,
        \[
        \Pr[\tau = \tau' \land \pw = \pw'] \geq 2^{-\gamma} + \epsilon - O(q_1^3/2^\lambda).
        \]
        $\Hyb_3$: In this hybrid, we move $\Sim.E$ to just after $\cA_1$ outputs.
        \begin{itemize}        
            \item $\pw \gets D$
            \item $(\sk,\pk,\phi) \gets \PAKE_1(\pw)$;
            \item $c,\tau \gets \cA_1^{H_1,H_2,\Sim.RO}(\phi)$;
            \item Compute $(\pw',\ast,\ast,\ast,\ast) = \Sim.E(\tau)$.
            \item Compute $K = \Dec(\sk,c), \tau' = \Sim.RO(\pw,\pk,\phi,c,K)$.
            \item Output $\pw,\pk,\phi,c,K,\tau, \tau',\pw'$.
        \end{itemize}
        By \zcref[S]{thm:simextract}, in $\Hyb_3$ \[
        \Pr[\tau = \tau' \land \pw = \pw'] \geq 2^{-\gamma} + \epsilon -  O(q_1^3/2^\lambda) - \frac{8\sqrt{2}}{2^{\lambda/2}}.
        \]
        Therefore if $\epsilon$ is non-negligible we have constructed an adversary guessing the password with non-negligible advantage: the adversary runs $c,\tau \gets \cA_1^{H_1,H_2,\Sim.RO}(\phi)$ (simulating $\Htag$ with $\Sim$), computes $(\pw',\ast,\ast,\ast,\ast) = \Sim.E(\tau)$ and outputs $\pw'$ as the password guess.
    \end{proof}
    
\section{Advantage-tight one-way to hiding}\label{sec:oth}

In this section, we establish a search-to-decision reduction in the quantum random oracle model where the search upper bound transfers to a distinguishing upper bound without advantage loss (allowing for polynomially greater running time).

In fact, we consider a very general scenario specified by any sampler $D$ that outputs $(H,G,S,\rho)$, where $H,G$ are functions from $\cX \to \cY$, $S$ is a subset of the domain $S \subseteq \cX$, and $\rho$ is a quantum state.\footnote{Technically we consider a family of samplers $H_\lambda,G_\lambda,S_\lambda,\rho_\lambda \gets D(1^\lambda)$ indexed by the security parameter, but this is left implicit. } We require that $H(x) = G(x) \ \forall x  \notin S$. Let $M_S : \cX \to \{0,1\}$ be the membership-checking predicate for $S$. For any (oracle-aided) distinguishing circuit $A$, we define its distinguishing advantage as \[d_A = \Big|\Pr_{(H,G,S,\rho) \gets D}[A^H(\rho) = 1] -  \Pr_{(H,G,S,\rho) \gets D}[A^G(\rho) = 1]\Big|,\] and for any (oracle-aided) search circuit $B$, we define its search advantage as \[p_B = \Pr_{(H,G,S,\rho) \gets D}[x \in S : x \gets B^{H,G,M_S}(\rho)].\] Notice that we allow the search algorithm to make oracle calls to the membership checking predicate $M_S$. 

We show that the distinguishing advantage is \emph{exactly} upper bounded by the search advantage (plus an additive negligible factor). We take all oracle operations to be unit cost, so below QPT means that the total size of all non-oracle operations is polynomial.

\begin{theorem}\label{thm:search-to-decision}
    Let $\epsilon : \mathbb{N} \to [0,1]$ and suppose that for every QPT $B$, $p_B \leq \epsilon(\lambda) + \mathsf{negl}(\lambda)$. Then for every QPT distinguisher $A$, $d_A \leq \epsilon(\lambda) + \mathsf{negl}(\lambda)$. 
\end{theorem}

The starting point for our proof is a decomposition of an arbitrary attacker's distinguishing advantage as the sum of two terms: a ``forward'' branch that sums its amplitudes on $S$ at each oracle query, and a ``reverse'' branch that, starting from the adversary's \emph{final} state, sums the amplitudes on $S$ at each oracle query when running the adversary in reverse. 

This decomposition is the central idea underlying the ``measure-rewind-measure'' technique of \cite{EC:KSSSS20}. However, their search bound loses a factor that depends on the number of queries $q$ made by the distinguisher, since their search adversary essentially chooses one (forward or reverse) query to measure at random.

Here, we move this factor of $q$ from the advantage bound to the \emph{running time} of the search adversary by using quantum rewinding. In particular, we define a unitary $U$ that runs the search adversary suggested by the above decomposition \emph{coherently} over the choice of branch $b \in \{\mathsf{forward},\mathsf{reverse}\}$ and query $t \in [q]$. We then define an ``input'' projector $\Pi_\inp$ that projects onto the clean starting space of this (purified) adversary, and a ``success'' projector $\Pi_\good$ that projects onto the set $S$. Then, by alternating the projectors $\Pi_\inp$ and $U^\dagger \Pi_\good U$ sufficiently many times, one can amplify \emph{almost all} of the adversary's input state that is responsible for its distinguishing advantage (which, by the decomposition stated above, lies in the span of right singular vectors of $\Pi_\good U \Pi_\inp$ above a certain threshold) onto the image of $\Pi_\good$. Measuring the output register then produces a string in $S$.

The quantum rewinding step is captured by the below amplitude amplification lemma, which follows readily from \cite[Theorem 26]{GSLV}.

\begin{lemma}[\cite{GSLV}]\label{lemma:amplification}
    Let $U$ be a unitary, let $\Pi_\inp$ and $\Pi_\good$ be projectors, and write \[\Pi_\good U \Pi_\inp = \sum_j \sigma_j\ketbra{u_j}{v_j}, \qquad P_{\geq \gamma} \coloneqq \sum_{j: \sigma_j \geq \gamma}\ketbra{v_j}{v_j}.\] For any $\gamma > 0$ and $\delta \in (0,1/2)$, there exists a unitary $\Amp_{\gamma,\delta}$ using a total of $O(\log(1/\delta)/\gamma)$ calls to $U, U^\dagger, \Pi_\inp, \Pi_\good$ such that for any $\rho$ supported on the image of $\Pi_\inp$, \[\Tr\left(\Pi_\good \Amp_{\gamma,\delta}(\rho)\right) \geq (1-\delta)^2\Tr\left(P_{\geq \gamma}\rho\right).\] 
\end{lemma}

Given this lemma, we are ready to prove \zcref[S]{thm:search-to-decision}.

\begin{proof} (of \zcref[S]{thm:search-to-decision})
    Let $A$ be any QPT distinguisher that makes $q$ oracle calls and has advantage $d_A$. 
    
    \paragraph{Notation.} Fix a sample $H,G,S,\rho$ from $D$ and consider a purified implementation of $A$ that operates on a register $\Zr$ that holds its input state $\rho$ along with any auxiliary registers required for the purification of the strategy. We also define a clock register $\Rr_\clock$ spanned by $\ket{0},\ket{1},\dots,\ket{q}$ and initialized to $\ket{0}$. Finally, let $\ket{\psi}_{\Zr,\Rr_\pure,\Rr_\clock}$ denote an arbitrary purification of $\rho$ using purification register $\Rr_\pure$, tensored with $\ket{0}_{\Rr_\clock}$. Note that $A$ never touches $\Rr_\pure$ or $\Rr_\clock$.

    Using the purified implementation of $A$, we write 
    \begin{align*}
        A_H &= A_q O_H A_{q-1} \dots A_1 O_H A_0 \\
        A_G &= A_q O_G A_{q-1} \dots A_1 O_G A_0,
    \end{align*}
     where $O_H$ is an oracle query to $H$ and $O_G$ is an oracle query to $G$. For $t \in [q]$, define 
    \begin{align*}
        F_t &= A_{t-1} O_G A_{t-2} \dots A_1 O_G A_0, \\
        R_t &= A_q O_H A_{q-1} \dots A_{t+1} O_H A_t.
    \end{align*}
    Next, define $W_q$ to be a unitary satisfying \[W_q\ket{0}_{\Rr_\clock} = \frac{1}{\sqrt{q}}\sum_{t = 1}^q\ket{t}_{\Rr_\clock},\]  define forward and reverse history-state unitaries \[U_\forward \coloneqq \left(\ketbra{0}_{\Rr_\clock} \otimes I + \sum_{t = 1}^q\ketbra{t}_{\Rr_\clock} \otimes F_t\right)\left(W_q \otimes I_{\Zr,\Rr_\pure}\right), \]\[\quad U_\rev \coloneqq \left(\ketbra{0}_{\Rr_\clock} \otimes I + \sum_{t = 1}^q\ketbra{t}_{\Rr_\clock} \otimes R_t^\dagger\right)\left(W_q \otimes I_{\Zr,\Rr_\pure}\right),\]
    and define projectors \[\Pi_\inp \coloneqq \ketbra{0}_{\Rr_\clock} \otimes I_{\Zr,\Rr_\pure}, \quad \Pi_\good \coloneqq \sum_{t = 1}^q\ketbra{t}_{\Rr_\clock} \otimes \Pi_S \otimes I_{\Rr_\pure},\] where $\Pi_S$ is the projection $\sum_{x \in S}\ketbra{x}$ applied to the query register of $A$. Note that for any state $\ket{\psi}$ in the image of $\Pi_\inp$, \begin{equation}\label{eq:forward-expand}
        \| \Pi_\good U_\forward\ket{\psi}\|^2 = \frac{1}{q}\sum_{t \in [q]}\|\Pi_S F_t\ket{\psi}\|^2
    \end{equation} 
    and
    \begin{equation}\label{eq:reverse-expand}
        \| \Pi_\good U_\rev\ket{\psi}\|^2 = \frac{1}{q}\sum_{t \in [q]}\|\Pi_S R_t^\dagger\ket{\psi}\|^2.
    \end{equation}

    Finally, let $\Or$ be the single-qubit sub-register of $Z$ on which $A$ writes its output bit, and define
    \[\Pi_\acc \coloneqq \ketbra{1}_\Or \otimes I.\]

    \paragraph{Approach.} As explained above, we will now construct a unitary $U_A$ based on the description of $A$ such that, for some threshold $r$, any right singular vector of $\Pi_\good U_A \Pi_\inp$ with singular value below $1/r$ contributes negligibly to the distinguishing advantage of $A$. Moreover, we want the complexity of implementing $U_A$ to be comparable to $A$. If this can be arranged, then we can apply \zcref[S]{lemma:amplification} to amplify the overwhelming majority of $A$'s distinguishing advantage into the subspace $\Pi_\good$, which produces a successful search answer.

    The idea for $U_A$ will be derived from re-writing the expression for $A$'s distinguishing advantage as follows. 

    \begin{align*}
        &\big|\Pr[A^H(\ket{\psi}) = 1] - \Pr[A^G(\ket{\psi}) = 1]\big| \\
        &\quad= \big|\bra{\psi}\left(A_H^\dagger\Pi_\acc A_H - A_G^\dagger\Pi_\acc A_G\right)\ket{\psi}\big| \\
        &\quad= \big|\bra{\psi}\left(A_H^\dagger\Pi_\acc A_H - A_H^\dagger \Pi_\acc A_G + A_H^\dagger \Pi_\acc A_G - A_G^\dagger\Pi_\acc A_G\right)\ket{\psi}\big| \\ 
        &\quad= \big|\bra{\psi}\left(A_H^\dagger \Pi_\acc\Delta + \Delta^\dagger \Pi_\acc A_G\right)\ket{\psi}\big| \\
        &\quad\leq \| \Delta\ket{\psi}\| + \|\Delta^\dagger \Pi_\acc A_G\ket{\psi}\|,
    \end{align*}

    where 

    \[\Delta \coloneqq A_H - A_G = \sum_{t \in [q]} R_t(O_H - O_G)F_t.\]

    We have that 
    \begin{align*}
        \|\Delta \ket{\psi}\| &\leq \sum_{t \in [q]}\|R_t (O_H - O_G)F_t\ket{\psi}\| \\
        &\leq 2\sum_{t \in [q]} \|\Pi_SF_t\ket{\psi}\| \\ &\leq 2\sqrt{q}\left(\sum_{t \in [q]}\|\Pi_S F_t\ket{\psi}\|^2\right)^{1/2} \\ &= 2\sqrt{q}\left(q\|\Pi_\good U_\forward\ket{\psi} \|^2\right)^{1/2}\\&= 2q\|\Pi_\good U_\forward\ket{\psi}\|,
    \end{align*}
    where the last inequality is Cauchy-Schwarz and the first equality is \zcref[S]{eq:forward-expand}. Similarly,  
    \begin{align*}
        \| \Delta^\dagger \Pi_\acc A_G \ket{\psi}\| &\leq \sum_{t \in [q]}\|F_t^\dagger (O_H - O_G)^\dagger R_t^\dagger \Pi_\acc A_G\ket{\psi}\| \\ &\leq 2\sum_{t \in [q]} \|\Pi_S R_t^\dagger \Pi_\acc A_G\ket{\psi}\| \\
        &\leq 2\sqrt{q}\left(\sum_{t \in [q]}\|\Pi_S R_t^\dagger \Pi_\acc A_G\ket{\psi} \|^2\right)^{1/2} \\
        &= 2q\|\Pi_\good U_\rev\Pi_\acc A_G\ket{\psi}\|,
    \end{align*}
    using \zcref[S]{eq:reverse-expand} and the fact that $A_G$ and $\Pi_\acc$ preserve $\Pi_\inp$.

    To summarize, we can write \begin{equation}\label{equation:decomposition}\begin{aligned}
        &\big|\Pr[A^H(\ket{\psi}) = 1] - \Pr[A^G(\ket{\psi}) = 1]\big| \\
        &\quad\leq \| \Delta\ket{\psi}\| + \|\Delta^\dagger \Pi_\acc A_G\ket{\psi}\| \\
        &\quad\leq 2q\left(\|\Pi_\good U_\forward\ket{\psi}\| + \|\Pi_\good U_\rev\Pi_\acc A_G\ket{\psi}\|\right).
    \end{aligned} \end{equation}
    
    In other words, up to the $2q$ scaling, one component of the distinguishing advantage comes from a ``forward'' branch $\Pi_\good U_\forward\ket{\psi}$ and the other comes from a ``reverse'' branch $\Pi_\good U_\rev\Pi_\acc A_G\ket{\psi}$. Similar observations have been used in \cite{EC:KSSSS20,AC:GeLiaXue24}.

    \paragraph{The unitary $U_A$.} We use this decomposition to inform the specification of the unitary $U_A$. First, expand the register $\Zr$ to include single-qubit registers $\Br$ and $\Fr$ initialized to $\ket{0}$. $\Br$ will be a ``branch'' register, allowing both the forward and reverse branches to be run coherently, and $\Fr$ will be a register that allows for a unitary implementation of $\Pi_\acc$, that is, \[U_\acc \coloneqq I_\Fr \otimes (I - \Pi_\acc) + X_\Fr \otimes \Pi_\acc.\] Redefine $\ket{\psi}$ to include $\ket{00}_{\Br\Fr}$, redefine $\Pi_\inp$ as \[\Pi_\inp \coloneqq \ketbra{00}_{\Br\Fr} \otimes \ketbra{0}_{\Rr_\clock} \otimes I,\] and redefine $\Pi_\good = \Pi_\good \otimes I_{\Br\Fr}$. Then, define \[U_A \coloneqq \left(\ketbra{0}_\Br \otimes I_\Fr \otimes U_\forward + \ketbra{1}_\Br \otimes (I_\Fr \otimes U_\rev)U_\acc(I_\Fr \otimes A_G)\right)(H_\Br \otimes I),\] where $H$ is the Hadamard gate. 

    \paragraph{The singular value cutoff.} Write a singular value decomposition \[\Pi_\good U_A \Pi_\inp = \sum_{j}\sigma_j\ketbra{u_j}{v_j},\] and, fixing an $r$, define \[P \coloneqq \sum_{j : \sigma_j \geq 1/(8r)} \ketbra{v_j}.\] 
    
    We can now establish the central claim of the proof, showing that $A$'s distinguishing advantage comes fully from the $P$-image of its input $\ket{\psi}$, up to error $q/r$.

    \begin{claim}\label{claim:cutoff}
        \[\big|\Pr[A^H(\ket{\psi}) = 1] - \Pr[A^G(\ket{\psi}) = 1]\big| \leq \|P\ket{\psi}\|^2 + \frac{q}{r}.\]
    \end{claim}

    \begin{proof}
    Let
    \[
        L \coloneqq \Pi_\good U_A\Pi_\inp
    \]
    and define
    \[
        E \coloneqq
        \ketbra{00}_{\Br\Fr}
        \otimes \ketbra{0}_{\Rr_\clock}
        \otimes
        \left(
            A_H^\dagger\Pi_\acc A_H
            -
            A_G^\dagger\Pi_\acc A_G
        \right).
    \]

    We first compare $E$ with $L$. Let $\ket{\phi}$ be an arbitrary state (not necessarily in the image of $\Pi_\inp$) and define $\ket{\phi_\inp} \coloneqq \Pi_\inp\ket{\phi}$.
     Define
    \begin{align*}
        f(\phi)
        &:=
        \left\|
            \Pi_\good U_\forward
            \ket{\phi_\inp}
        \right\|,\\
        b(\phi)
        &:=
        \left\|
            \Pi_\good U_\rev
            \Pi_\acc A_G\ket{\phi_\inp}
        \right\|.
    \end{align*}

    Then 
    \begin{align*} 
        \|E\ket{\phi}\| &= \|\left(
            A_H^\dagger\Pi_\acc A_H
            -
            A_G^\dagger\Pi_\acc A_G
        \right)\ket{\phi_\inp}\| \\ &\leq \| \Delta\ket{\phi_\inp}\| + \|\Delta^\dagger \Pi_\acc A_G\ket{\phi_\inp}\| \\ &\leq 2q\bigl(f(\phi)+b(\phi)\bigr),
    \end{align*}

    using \zcref[S]{equation:decomposition} applied to the possibly
    unnormalized vector $\ket{\phi_\inp}$.

    By the definition of $U_A$, we have
    \begin{align*}
        L\ket{\phi}
        =\frac{1}{\sqrt{2}}\Bigl(
        &\ket{00}_{\Br\Fr}\,
        \Pi_\good U_\forward\ket{\phi_\inp}\\
        {}+{}&
        \ket{10}_{\Br\Fr}\,
        \Pi_\good U_\rev(I-\Pi_\acc)A_G\ket{\phi_\inp}\\
        {}+{}&
        \ket{11}_{\Br\Fr}\,
        \Pi_\good U_\rev\Pi_\acc A_G\ket{\phi_\inp}
        \Bigr),
    \end{align*}
    so discarding the second summand we have
    \[
        \|L\ket{\phi}\|^2
        \geq
        \frac{1}{2}\bigl(f(\phi)^2+b(\phi)^2\bigr).
    \] Thus, for any input state $\ket{\phi}$,

    \[\| E\ket{\phi}\| \leq 2q\left(f(\phi) + b(\phi)\right) \leq 2q\sqrt{2}\sqrt{f(\phi)^2 + b(\phi)^2} \leq 4q\| L\ket{\phi}\|.\]

    This establishes that \[
        \|E(I-P)\|
        \leq
        4q\|L(I-P)\|
        \leq
        \frac{q}{2r},\]
    where the second inequality follows by definition of $P$ using threshold $\gamma = 1/(8r)$. Therefore, \[
        \|E-PEP\| = \|(I-P)E + P E (I-P)\|  \leq \|(I-P)E\| + \|E (I-P)\|  \leq \frac{q}{r},\] using the fact that $E$ is Hermitian.

    Finally, since $\ket\psi$ is supported on the image of $\Pi_\inp$,
    \begin{align*}
        &\left|
            \Pr[A^H(\ket\psi)=1]
            -
            \Pr[A^G(\ket\psi)=1]
        \right|\\
        &\qquad=
        \left|\bra\psi E\ket\psi\right|\\
        &\qquad\leq
        \left|\bra\psi PEP\ket\psi\right|
        +
        \|E-PEP\|\\
        &\qquad\leq
        \bra\psi P\ket\psi
        +
        \frac{q}{r}\\
        &\qquad=
        \|P\ket{\psi}\|^2
        +
        \frac{q}{r},
    \end{align*}
    where the second inequality uses the fact that, since $E$ is the difference of two projectors, $-I \preceq E \preceq I$.



    

    \end{proof}

    \paragraph{Amplifying and completing the proof.} Apply \zcref[S]{lemma:amplification} with $\Pi_\good, U = U_A, \Pi_\inp, \gamma = 1/(8r)$, and $\delta = 2^{-\lambda}$. This yields an amplified unitary $\Amp$ such that \[\|\Pi_\good\Amp\ket{\psi}\|^2 \geq (1-\mathsf{negl}(\lambda))\|P\ket{\psi}\|^2.\] 
    
    Define the search adversary $B$ to run $\Amp$ on the expanded register $\Zr$ and $\Rr_\clock$ and then measure the query address to obtain a response $x$. $B$ can be implemented using $O(\lambda r)$ calls to $U_A,U_A^\dagger, \Pi_\good$, and $\Pi_\inp$, and thus with the same number of calls to $A$ and $\Pi_S$. 
    
    Hence, taking the expectation over $(H,G,S,\rho) \gets D$, and applying \zcref[S]{claim:cutoff}, we have that 
    \begin{align*}
        d_A &= \Big|\E_{(H,G,S,\rho) \gets D}\left[\Pr[A^H(\ket{\psi}) = 1] - \Pr[A^G(\ket{\psi}) = 1]\right]\Big| \\
        &\leq \E_{(H,G,S,\rho) \gets D}\Big| \Pr[A^H(\ket{\psi}) = 1] - \Pr[A^G(\ket{\psi}) = 1]\Big| \\ &\leq \E_{(H,G,S,\rho) \gets D}\|P\ket{\psi}\|^2 + \frac{q}{r} \\
        &\leq \frac{p_B}{(1-\mathsf{negl}(\lambda))} + \frac{q}{r} \\ &\leq \epsilon(\lambda) + \frac{q}{r} + \mathsf{negl}(\lambda).
    \end{align*}

    Finally, assume for contradiction that there exists a polynomial $s(\lambda)$ such that on infinitely many security parameters, $d_A > \epsilon(\lambda) + \frac{1}{s(\lambda)}$. Then setting $r = 2q s(\lambda) = \mathsf{poly}(\lambda)$ we obtain a contradiction, completing the proof.

\end{proof}

Now we use \zcref[S]{thm:search-to-decision} to analyze the following scenario. Let $(k,\rho) \gets D$ be an arbitrary sampler that outputs $k \in \cX$ and a state $\rho$. Let $H : \cX \to \cY$ be a uniformly random function with $|\cY| = \lambda^{\omega(1)}$, sampled independently of $D$. We have the following corollary.

\begin{corollary}\label{cor:search-to-decision}
    Let $\epsilon: \mathbb{N} \to [0,1]$ and suppose that for every QPT $B$, \[\Pr_{(k,\rho) \gets D,H}\left[B^H(\rho,H(k)) = k\right] \leq \epsilon(\lambda) + \mathsf{negl}(\lambda).\] Then for every QPT $A$, \[\bigg| \Pr_{(k,\rho) \gets D, H}\left[A^H(\rho,H(k)) = 1\right] - \Pr_{(k,\rho) \gets D, H, u \gets \cY}\left[A^H(\rho,u) = 1\right]\bigg| \leq \epsilon(\lambda) + \mathsf{negl}(\lambda).\] 
\end{corollary}

\begin{proof}

    First, given $k$ and $H$, define 
    \[y \coloneqq H(k), \quad S \coloneqq H^{-1}(y), \quad G(x) \coloneqq \begin{cases}a \ \ \ \ \ \ \text{if } H(x) = y \\ H(x) \ \ \text{otherwise}\end{cases},\] where $a \gets \cY$. Then since $G$ and $M_S$ are efficiently implementable given $y,a$, and oracle access to $H$, and $H$ is collision-resistant whenever $|\cY| = \lambda^{\omega(1)}$, we have that \[\Pr_{(k,\rho) \gets D, H, G}[B^{H,G,M_S}(\rho,y) \in S \setminus \{k\}] = \mathsf{negl}(\lambda),\] where $G$ is sampled as described above. Combining with the assumption in the corollary statement, we get \[\Pr_{(k,\rho) \gets D,H,G}[B^{H,G,M_S}(\rho,y) \in S] \leq \epsilon(\lambda) + \mathsf{negl}(\lambda).\] So by \zcref[S]{thm:search-to-decision}, we have that \[\bigg| \Pr_{(k,\rho) \gets D, H, G}\left[A^H(\rho,H(k)) = 1\right] - \Pr_{(k,\rho) \gets D, H, G}\left[A^G(\rho,H(k)) = 1\right]\bigg| \leq \epsilon(\lambda) + \mathsf{negl}(\lambda).\]

    Now, define \[\widetilde{G}(x) \coloneqq \begin{cases}a, \ \ \ \ \ \ \text{if }x = k \\ H(x) \ \ \text{otherwise}\end{cases},\] and note that $G$ and $\widetilde{G}$ only differ on $S \setminus \{k\}$. Again due to collision-resistance of $H$ and \zcref[S]{thm:search-to-decision} with $\epsilon(\lambda) = 0$ (though here standard one-way to hiding would suffice), we have that \[\bigg| \Pr_{(k,\rho) \gets D, H, G}\left[A^G(\rho,H(k)) = 1\right] - \Pr_{(k,\rho) \gets D, H, \widetilde{G}}\left[A^{\widetilde{G}}(\rho,H(k)) = 1\right]\bigg| \leq \mathsf{negl}(\lambda).\] Finally, note that the distribution of $k,\rho,H(k),\widetilde{G}$ is \emph{identical} to the distribution of $k,\rho,u,H$ in the experiment in the statement of the corollary. This completes the proof.

\end{proof}

\section{Acknowledgments}

JB would like to thank Saikrishna Badrinarayan, Daniel Masny, and Pratyay Mukherjee for a previous collaboration related to the post-quantum security of the Masny-Rindal OT protocol, leading to insights that guided the direction of this work. JJ would like to thank Lawrence Roy and Naman Kumar for helpful discussions.

\paragraph{AI usage.} The main contributions of this work were developed without AI assistance. AI tools were used to find the central idea behind the proof of \zcref[S]{thm:search-to-decision} (advantage-tight one-way to hiding) and to aid with proofreading, checking of technical details, and polishing.

\bibliographystyle{alpha}
\bibliography{abbrev3,crypto,main}

@string{ieee =                  {IEEE}}

@string{springer =              "Springer"}

@string{acm =                   "{ACM}"}

@article{grinko2025quantum,
  title={Quantum Simulation of Random Unitaries from Clebsch-Gordan Transforms},
  author={Grinko, Dmitry and Yoshida, Satoshi},
  journal={arXiv preprint arXiv:2509.26623},
  year={2025}
}

@article{foxman2026quantum,
  title={Quantum Lazy Sampling and Path Recording for Any Group},
  author={Foxman, Ben and Lombardi, Alex and Ma, Fermi and Nehoran, Barak and Wright, John},
  journal={arXiv preprint arXiv:2606.30281},
  year={2026}
}

@inproceedings{carolan2026compressed,
  title={Compressed permutation oracles},
  author={Carolan, Joseph},
  booktitle={Proceedings of the 58th Annual ACM Symposium on Theory of Computing},
  pages={150--161},
  year={2026}
}

@misc{cryptoeprint:2026/1331,
  author       = {Jiayu Xu},
  title        = {The Most Efficient Protocol for {PAKE}: What Exact Stuff Do You Need to Hash at the End?},
  howpublished = {Cryptology {ePrint} Archive, Paper 2026/1331},
  year         = {2026},
  url          = {https://eprint.iacr.org/2026/1331}
}

@inproceedings{masnyrindal,
author = {Mansy, Daniel and Rindal, Peter},
title = {Endemic Oblivious Transfer},
year = {2019},
isbn = {9781450367479},
publisher = {Association for Computing Machinery},
address = {New York, NY, USA},
url = {https://doi.org/10.1145/3319535.3354210},
doi = {10.1145/3319535.3354210},
booktitle = {Proceedings of the 2019 ACM SIGSAC Conference on Computer and Communications Security},
pages = {309–326},
numpages = {18},
location = {London, United Kingdom},
series = {CCS '19}
}

@inproceedings{dall2023necessity,
  title={On the Necessity of Collapsing for Post-Quantum and Quantum Commitments},
  author={Dall’Agnol, Marcel and Spooner, Nicholas},
  booktitle={18th Conference on the Theory of Quantum Computation, Communication and Cryptography},
  year={2023}
}

@article{hovelmanns2025cake,
  title={CAKE requires programming-on the provable post-quantum security of (O) CAKE},
  author={H{\"o}velmanns, Kathrin and H{\"u}lsing, Andreas and Kudinov, Mikhail and Ritsch, Silvia},
  journal={Cryptology ePrint Archive},
  year={2025}
}

@article{arriaga2025noic,
  title={NoIC: PAKE from KEM without ideal ciphers},
  author={Arriaga, Afonso and Barbosa, Manuel and Jarecki, Stanislaw},
  journal={Cryptology ePrint Archive},
  year={2025}
}

@article{arriaga2026tempo,
  author       = {Afonso Arriaga and Manuel Barbosa and Stanislaw Jarecki},
  title        = {Tempo: An ML-KEM to PAKE Compiler Resilient to Timing Attacks},
  journal      = {IACR Transactions on Cryptographic Hardware and Embedded Systems},
  year         = {2026},
}

@article{vos2025hybrid,
  title={A hybrid asymmetric password-authenticated key exchange in the random oracle model},
  author={Vos, Jelle and Jarecki, Stanislaw and Wood, Christopher A and Yun, Cathie and Myers, Steve and Sierra, Yannick},
  journal={Cryptology ePrint Archive},
  year={2025}
}

@misc{CFRG_PAKE_Selection,
  author       = {{Crypto Forum Research Group}},
  title        = {{CFRG PAKE selection}},
  year         = {2019},
  howpublished = {\url{https://github.com/cfrg/pake-selection}},
  note         = {Accessed: 2026-09-30}
}

@techreport{ETSI,
  author      = {{European Telecommunications Standards Institute}},
  title       = {Impact of Quantum Computing on Cryptographic Security Proofs},
  institution = {ETSI},
  type        = {ETSI Technical Report},
  number      = {TR 103 965 V1.1.1},
  year        = {2024},
  month       = dec,
  url         = {https://www.etsi.org/deliver/etsi_tr/103900_103999/103965/01.01.01_60/tr_103965v010101p.pdf}
}

@misc{vos_hybrid_draft,
  title        = {Hybrid Post-Quantum Password Authenticated Key Exchange},
  author       = {Vos, Jelle and Jarecki, Stanislaw and Wood, Christopher A.},
  howpublished = {\url{https://datatracker.ietf.org/doc/draft-vos-cfrg-pqpake/}},
  note         = {Internet-Draft, IETF},
  year         = {2026},
  urldate      = {2026-07-10}
}

@article{verschoor2022quantum,
  title={Quantum information in security protocols},
  author={Verschoor, Sebastian Reynaldo},
  year={2022},
  publisher={University of Waterloo}
}

@inproceedings{GSLV,
author = {Gily{\'e}n, Andr{\'a}s and Su, Yuan and Low, Guang Hao and Wiebe, Nathan},
title = {Quantum singular value transformation and beyond: exponential improvements for quantum matrix arithmetics},
year = {2019},
isbn = {9781450367059},
publisher = {Association for Computing Machinery},
address = {New York, NY, USA},
url = {https://doi.org/10.1145/3313276.3316366},
doi = {10.1145/3313276.3316366},
booktitle = {Proceedings of the 51st Annual ACM SIGACT Symposium on Theory of Computing},
pages = {193–204},
numpages = {12},
location = {Phoenix, AZ, USA},
series = {STOC 2019}
}

@inproceedings{LombardiMQW22,
  author    = {Alex Lombardi and Ethan Mook and Willy Quach and Daniel Wichs},
  title     = {Post-Quantum Insecurity from {LWE}},
  booktitle = {Theory of Cryptography -- TCC 2022, Part I},
  editor    = {Eike Kiltz and Vinod Vaikuntanathan},
  series    = {Lecture Notes in Computer Science},
  volume    = {13747},
  pages     = {3--32},
  publisher = {Springer},
  year      = {2022},
  doi       = {10.1007/978-3-031-22318-1_1},
  url       = {https://eprint.iacr.org/2022/869}
}

@inproceedings{BrakerskiCMVV18,
  author    = {Zvika Brakerski and Paul Christiano and Urmila Mahadev
               and Umesh Vazirani and Thomas Vidick},
  title     = {A Cryptographic Test of Quantumness and Certifiable
               Randomness from a Single Quantum Device},
  booktitle = {2018 IEEE 59th Annual Symposium on Foundations
               of Computer Science (FOCS)},
  pages     = {320--331},
  publisher = {IEEE},
  year      = {2018},
  doi       = {10.1109/FOCS.2018.00038},
  url       = {https://arxiv.org/abs/1804.00640}
}

@inproceedings{CojocaruHLYY25,
  author    = {Alexandru Cojocaru and Minki Hhan and Qipeng Liu
               and Takashi Yamakawa and Aaram Yun},
  title     = {Quantum Lifting for Invertible Permutations and Ideal Ciphers},
  booktitle = {Advances in Cryptology -- {CRYPTO} 2025, Part II},
  editor    = {Yael Tauman Kalai and Seny F. Kamara},
  series    = {Lecture Notes in Computer Science},
  volume    = {16001},
  pages     = {481--512},
  publisher = {Springer},
  year      = {2025},
  doi       = {10.1007/978-3-032-01878-6_16},
  url       = {https://eprint.iacr.org/2025/738}
}

@misc{carolan2026compressedpermutationoraclesrevisited,
      title={Compressed Permutation Oracles Revisited}, 
      author={Joseph Carolan and Christian Majenz},
      year={2026},
      eprint={2609.28469},
      archivePrefix={arXiv},
      primaryClass={quant-ph},
      url={https://arxiv.org/abs/2609.28469}, 
}

\appendix
\section{Generalizing the attack}\label{subsec:attack-general}
For the sake of clarity and concreteness, \zcref[S]{thm:IC-attack} shows an attack against a specific instantiation of ideal cipher-based OEKE. However, our attack is quite general, and applies \emph{mutatis mutandis} to all known variants of EKE and OEKE.

\paragraph{EKE.}
\begin{figure}[ht]
\pseudocodeblock{
\underline{P_1, \role = \text{``initiator''}} \< \< \underline{P_2, \role = \text{``respondent''}} \\ 
(\pk,\sk) \gets \KG \< \< \\
\< \sendmessageright*{\phi := \ICEnc_1(\pw,pk)} \< \\
\< \< \pk := \ICDec_1(\pw,\phi) \\
\< \< (c',K') := \Enc(\pk) \\
\< \sendmessageleft*{\phi' = \ICEnc_2(\pw,c')} \< \\
c = \ICDec_2(\pw,\phi') \\
K := \Dec(\sk,c) \< \< \\
\pkey = H(\pw, \pk, \phi,\phi', c,K) \< \< \\
\text{ output }\pkey \< \< \text{output }\pkey' = H(\pw, \pk, \phi, \phi', c',K') \\
}
\caption{A typical instantiation of EKE using a KEM $\KEM$, public key ideal cipher $\ICEnc_1, \ICDec_1$ and ciphertext ideal cipher $\ICEnc_2, \ICDec_2$.}
\label{EKEdiag}
\end{figure}
Encrypted Key Exchange (shown in \zcref[S]{EKEdiag}) uses the same first message as OEKE, but instead of sending an authentication tag in the second message, the KEM ciphertext is also encrypted under the password using an ideal cipher. Since our attack only requires the environment to send a first message, it also works against EKE: the only difference is that $\Env_1$ must now query $\ICDec_2(\pw_i, \phi')$ to obtain the KEM ciphertext for key derivation.

\paragraph{2-Feistel/Tempo.}
Modern variants of OEKE do not use a full ideal cipher: the most popular alternative is a 2-Feistel network (also known as a Programmable Once Public Function) \cite{CCS:McQRosRoy20,EC:JanRoyXu25,arriaga2025noic}.
\\
Let $H_1 : \mathcal{PW} \times \{0,1\}^\ell \to \PK, H_2: \mathcal{PW} \times \PK \to \{0,1\}^\ell$ be random oracles, where $\PK$ is the space of public keys and an abelian group with operation $\odot$, and $\mathcal{PW}$ is the space of passwords ($\ell = \omega(\log \lambda)$; one can take $\ell = \lambda$ for simplicity). We define the 2-Feistel based encryption and decryption algorithms as follows.

\begin{itemize}
    \item $\cE_\TF(\pw,\pk;w) = (s,T)$, where $R = H_1(\pw,w)$, $T = pk \odot R$, $m = H_2(\pw,T)$, $s = m \oplus w.$
    \item $\cD_\TF(\pw,(s,T)) = \pk$, where $ m = H_2(\pw, T)$, $w = s \oplus m$, $\pk = H_1(\pw, w)^{-1} \odot T.$
\end{itemize}
One can easily check that $\cD_\TF(\pw,\cE_\TF(\pw,\pk;w)) = \pk$ ($w$ is recovered by the computation of $\cD_\TF$ as well).

The 2-Feistel is a more efficient alternative to an ideal cipher, since in general one needs eight rounds of a Feistel network to be indifferentiable from an IC \cite{C:DaiSte16}. Conversely, the 2-Feistel is a weaker primitive than an IC:
\begin{enumerate}
    \item Encryption now uses additional randomness so there are many possible encryptions of $pk$ under a given $\pw$;
    \item 2F ciphertexts are partially malleable: given a ciphertext $(s,T)$ that is an encryption of $pk$ under $\pw$, an adversary can produce $(s',T')$ that is an encryption of $pk \odot \Delta$ under $\pw$ for adversarially chosen $\Delta$.
\end{enumerate}
These differences make the security proof more complex but still possible (see  \cite{EC:JanRoyXu25,arriaga2025noic} for details). 
\\
The Tempo protocol \cite{arriaga2026tempo} instantiates the above 2-Feistel OEKE template to use ML-KEM. The only difference is it splits public keys into two pieces $t \in \PK_1,\zeta \in \{0,1\}^{\ell'}$ via an efficient and invertible procedure $\mathsf{Split}$, where $\zeta$ is uniformly random when $\pk$ is honestly generated. The protocol initiator 2-Feistel encrypts $t$ only, sending $\zeta$ in the clear; the respondent decrypts and applies $\mathsf{Split}^{-1}$ to get a public key. Additionally, the security argument relies on a $\mathsf{KeyGen2}$ function that generates $(\sk,\pk)$ from $\zeta$: when $\zeta$ is chosen uniformly $\mathsf{KeyGen2}(\zeta)$ is guaranteed to be distributed identically to an honest key pair.
\begin{theorem}
    Let $\Pi$ be an OEKE protocol instantiated with (1) any KEM with negligible correctness error (\zcref[S]{def:KEM-correctness}), (2) the 2-Feistel-based encryption and decryption algorithms $\cE_\TF,\cD_\TF$ defined above, or the Tempo algorithms, (3) any key derivation function $\cK$, and (4) any tag function $\cT$. Then $\Pi$ does not post-quantum UC-realize $\FPake$.
\end{theorem}

\begin{proof}

Consider the following environment interacting with 2-Feistel OEKE;
\paragraph{Environment $\Env_0$}
\begin{enumerate}
    \item Sample $b \gets \{0,1\}$ and send $(\NewS, \sid, P_2, P_1, \pw_b)$ to $P_2$ for some nonce $\sid$.
    \item Prepare the normalized state
    \[
    \frac{1}{\sqrt{2|\PK|2^\ell}}\sum_{\substack{u \in \{0,1\}, pk \in \PK,\\ w \in \{0,1\}^\ell}} \ket{w}_{\Cr}\ket{\pw_u}_{\Ar}\ket{\pw_u}_{\Xr}\ket{pk}_{\Br}\ket{pk}_{\Yr}.
    \]
    \item Query $H_1$ on registers $\Xr, \Cr$ into a new register $\Zr_1$ (through $\adv$) to compute \[
    \begin{aligned}
    &\frac{1}{\sqrt{2|\PK|2^\ell}}\sum_{\substack{u \in \{0,1\}, pk \in \PK,\\ w \in \{0,1\}^\ell}} \ket{w}_{\Cr}\ket{\pw_u}_{\Ar}\ket{\pw_u}_{\Xr}\ket{pk}_{\Br}\ket{pk}_{\Yr}\\
    &\qquad\ket{pk \odot H_1(\pw_u, w)}_{\Zr_1}.
    \end{aligned}
    \]
    \item Query $H_2$ on registers $\Xr, \Zr_1$ into a new register $\Zr_2$ (through $\adv$) to compute \[
    \begin{aligned}
    &\frac{1}{\sqrt{2|\PK|2^\ell}}\sum_{\substack{u \in \{0,1\}, pk \in \PK,\\ w \in \{0,1\}^\ell}} \ket{w}_{\Cr}\ket{\pw_u}_{\Ar}\ket{\pw_u}_{\Xr}\ket{pk}_{\Br}\ket{pk}_{\Yr}\\
    &\qquad\ket{pk \odot H_1(\pw_u, w)}_{\Zr_1}\ket{w \oplus H_2(\pw_u, \Zr_1)}_{\Zr_2}.
    \end{aligned}
    \]
    \item Measure $\Zr_2, \Zr_1$ to get $(s,T)$ and leftover state $\rho^0_{\Cr \Ar \Br \Xr \Yr}$.
    \item For each $u \in \{0,1\}$, compute \[
    m_u = H_2(\pw_u, T), w_u = s \oplus m_u, \pk_u = H_1(\pw_u, w_u)^{-1} \odot T.
    \]
    \item Let $\ket{\psi} = \frac{1}{\sqrt{2}} \sum_{u \in \{0,1\}} \ket{w_u}_{\Cr}\ket{\pw_u}_{\Ar}\ket{\pw_u}_{\Xr}\ket{\pk_u}_{\Br}\ket{\pk_u}_{\Yr}$ and apply the measurement $M = \{\ket{\psi}\bra{\psi}, I - \ket{\psi}\bra{\psi}\}$ to $\rho^0$.
    \item Output $1$ if $M$ accepts: otherwise output $0$.
\end{enumerate}
The above $\Env_0$ works exactly the same as for ideal cipher OEKE: the key point is that given a candidate password $\pw$, decryption of $(s,T)$ under $\pw$ not only recovers the public key but also the randomness $w$ used to encrypt, so the post-measurement state $\ket{\psi}$ can be computed. Put another way, the computation of the first message is classically reversible given $\pw$. $\Env_1$ is modified in the same way: since the only difference between $\Env_0$ and $\Env_1$ through Step 6 is how public keys are sampled, we can again reduce to public key uniformity.
\\
The case of Tempo is analogous:
\paragraph{Environment $\Env_0$}
\begin{enumerate}
    \item Sample $b \gets \{0,1\}$ and send $(\NewS, \sid, P_2, P_1, \pw_b)$ to $P_2$ for some nonce $\sid$.
    \item Prepare the normalized state  \[ \frac{1}{\sqrt{2|\PK_1|2^{\ell+\ell'}}}\sum_{\substack{u \in \{0,1\}, t \in \PK_1, \\ w \in \{0,1\}^\ell, \zeta \in \{0,1\}^{\ell'}}} \ket{w}_{\Cr}\ket{\zeta}_{\Wr}\ket{\pw_u}_{\Ar}\ket{\pw_u}_{\Xr}\ket{t}_{\Br}\ket{t}_{\Yr}. \]
    \item Query $H_1$ on registers $\Xr, \Wr, \Cr$ into a new register $\Zr_1$ (through $\adv$) to compute \[
    \begin{aligned}
    &\frac{1}{\sqrt{2|\PK_1|2^{\ell+\ell'}}}\sum_{\substack{u \in \{0,1\}, t \in \PK_1, \\ w \in \{0,1\}^\ell, \zeta \in \{0,1\}^{\ell'}}} \ket{w}_{\Cr}\ket{\zeta}_{\Wr}\ket{\pw_u}_{\Ar}\ket{\pw_u}_{\Xr}\ket{t}_{\Br}\ket{t}_{\Yr}\\
    &\qquad\ket{H_1(\pw_u, \zeta, w) \odot t}_{\Zr_1}.
    \end{aligned}
    \]
    \item Query $H_2$ on registers $\Xr, \Wr, \Zr_1$ into a new register $\Zr_2$ (through $\adv$) to compute \[
    \begin{aligned}
    &\frac{1}{\sqrt{2|\PK_1|2^{\ell+\ell'}}}\sum_{\substack{u \in \{0,1\}, t \in \PK_1, \\ w \in \{0,1\}^\ell, \zeta \in \{0,1\}^{\ell'}}} \ket{w}_{\Cr}\ket{\zeta}_{\Wr}\ket{\pw_u}_{\Ar}\ket{\pw_u}_{\Xr}\ket{t}_{\Br}\ket{t}_{\Yr}\\
    &\qquad\ket{H_1(\pw_u, \zeta, w) \odot t}_{\Zr_1}\ket{w \oplus H_2(\pw_u, \zeta, \Zr_1)}_{\Zr_2}.
    \end{aligned}
    \]
    \item Measure $\Zr_2, \Zr_1, \Wr$ to get $s,T, \zeta$ and leftover state $\rho^0_{\Cr \Ar \Br \Xr \Yr}$.
    \item For each $u \in \{0,1\}$, compute \[
    m_u = H_2(\pw_u, \zeta, T), w_u = s \oplus m_u, t_u = H_1(\pw_u, \zeta, w_u)^{-1} \odot T.
    \]
    \item Let $\ket{\psi} = \frac{1}{\sqrt{2}} \sum_{u \in \{0,1\}} \ket{w_u}_{\Cr}\ket{\pw_u}_{\Ar}\ket{\pw_u}_{\Xr}\ket{t_u}_{\Br}\ket{t_u}_{\Yr}$ and apply the measurement $M = \{\ket{\psi}\bra{\psi}, I - \ket{\psi}\bra{\psi}\}$ to $\rho^0$.
    \item Output $1$ if $M$ accepts: otherwise output $0$.
\end{enumerate}
In $\Env_1$, we still prepare a superposition over all $\zeta$, but instead compute $t$ via $t = \mathsf{Split}(\mathsf{KeyGen2}(1^\lambda, \zeta)[1])[0]$. Thus we rely on the assumption that when $\zeta$ is chosen uniformly, an efficient QPT adversary cannot distinguish $(t \gets \PK_1, \zeta \in \{0,1\}^{\ell'}$ and $\zeta \in \{0,1\}^{\ell'}$ with $t$ as above. This is implied by the $\mathsf{UNI-PK^+}$ assumption in \cite{arriaga2026tempo}, and is necessary for the post-quantum security of Tempo by the same logic as the necessity of public key uniformity for OEKE: otherwise an adversary can mount an offline attack by decrypting the first message under many candidate passwords.

\end{proof}

\section{Modified Measure-and-Reprogram}\label{section:newmandrp}
We sketch why $V$ can be given access to $H$ in \zcref[S]{thm:mandp}. First we introduce some notation from \cite{C:DonFehMaj20}.  Let $\ket{\phi_i}$ be defined as $\cA$’s state right before making its $i+1$st query—with the special case $\ket{\phi_q}$ denoting the final output state—
to which we add the superscript $\mathcal{O}$ when all previous queries have been answered using $\mathcal{O}$.
Next, we use $\cA^\mathcal{O}_{i \to j}$
to denote the unitary that brings $\cA$ from $\ket{\phi_i}$ to $\ket{\phi_j}$, using $\mathcal{O}$ from the $i$-th query on.
Finally, we use the shorthand $X \coloneqq \ket{x}\bra{x}$. The measure-and-reprogram proof starts with the following lemma:
\begin{lemma}
    Let $\cA$ be a $q$-query oracle quantum algorithm. Then, for any function $H: \mathcal{X} \to \mathcal{Y}$, any $x \in \mathcal{X}, \Theta \in \mathcal{Y}$, and any projector $\Pi_{x,\Theta}$, it holds that
\[
\mathbb{E}_{i,b}
\left[
\left\|
\left(
X \otimes \Pi_{x,\Theta}
\right)
\left(
\mathcal{A}^{H(x \ast \Theta)}_{i+b \to q}
\right)
\left(
\mathcal{A}^{H}_{i \to i+b}
\right)
X
|\phi_i^H\rangle
\right\|^2
\right]
\geq
\frac{
\left\|
\left(
X \otimes \Pi_{x,\Theta}
\right)
|\phi_q^{H(x  * \Theta)}\rangle
\right\|^2
}{
(2q+1)^2
}.
\]
\end{lemma}
Here the expectation is over a random pair $(i,b) \in (\{0, \dots, q-1\} \times \{0,1\}) \cup (q,0)$. The projector $\Pi_{x,\Theta}$ will function as $V$ later on. Notice that the expectation is only over the randomness of the measure-and-reprogram simulator, and it holds for any \emph{fixed} $x,\Theta,\Pi_{x,\Theta},H$. Therefore the same lemma holds for a projector $\Pi_{x,\Theta}^{H(x \ast \Theta)}$ that gets access to the reprogrammed random oracle.
\\
\\
The next lemma is proven by induction using the previous lemma as a base case:
\begin{lemma}\label{mpint}
Let $\cA$ be a $q$-query oracle quantum algorithm. Then, for any function $H: \mathcal{X} \to \mathcal{Y}$, any $\mathbf{x} \in \mathcal{X}^n, \mathbf{\Theta} \in \mathcal{Y}^n$, and any projector $\Pi_{\mathbf{x},\mathbf{\Theta}}$, it holds that
    \[
\frac{
    \left\|
        \left(
            |\mathbf{x}\rangle\langle\mathbf{x}|
            \otimes \Pi_{\mathbf{x},\boldsymbol{\Theta}}
        \right)
        \mathcal{A}^{H * \boldsymbol{\Theta}_{\mathbf{x}}}
        |\phi_0\rangle
    \right\|^2
}{
    (2q+1)^{2n}
}
\leq
\mathbb{E}_{\mathbf{r}}
\left[
    \left\|
        \left(
            |\mathbf{x}\rangle\langle\mathbf{x}|_{A}
            \otimes \Pi_{\mathbf{x},\boldsymbol{\Theta}}
        \right)
        S_{\mathbf{r}}^{H}(\mathcal{A})
        |\phi_0\rangle
    \right\|^2
\right].
\]
\end{lemma}
The notation required to define $S_{\mathbf{r}}^{H}(\mathcal{A})$ formally is quite complex so we omit it, only noting that $S_{\mathbf{r}}^{H}(\mathcal{A})$ is the simulator that chooses $n$ queries $\mathbf{r} = (r_1, \dots, r_n)$ of $\cA$ to $H$ to measure and reprogram. The key observation is that the lemma once again holds for any fixed $H,\mathbf{x},\mathbf{\Theta},\Pi_{\mathbf{x},\mathbf{\Theta}}$ so it also holds for a projector $\Pi_{\mathbf{x},\mathbf{\Theta}}^{H(\mathbf{x} \ast \mathbf{\Theta})}$.

Finally, we come to the proof of \cite[Theorem~6]{C:DonFehMaj20}, the statement of which is the same as \zcref[S]{thm:mandp} but $V$ is not given oracle access. The proof proceeds by taking the expectation of \zcref[S]{mpint} over random $H, \mathbf{\Theta}$
\[
\begin{aligned}
&\mathbb{E}_{H, \mathbf{\Theta}}\frac{
    \left\|
        \left(
            |\mathbf{x}\rangle\langle\mathbf{x}|
            \otimes \Pi_{\mathbf{x},\boldsymbol{\Theta}}
        \right)
        \mathcal{A}^{H * \boldsymbol{\Theta}_{\mathbf{x}}}
        |\phi_0\rangle
    \right\|^2
}{
    (2q+1)^{2n}
}
\\ &\quad\leq
\mathbb{E}_{\mathbf{r},H, \mathbf{\Theta}}
\left[
    \left\|
        \left(
            |\mathbf{x}\rangle\langle\mathbf{x}|_{A}
            \otimes \Pi_{\mathbf{x},\boldsymbol{\Theta}}
        \right)
        S_{\mathbf{r}}^{H}(\mathcal{A})
        |\phi_0\rangle
    \right\|^2
\right]
\end{aligned}
\]
and noting that after a change of variables the left hand side equals 
\[
\mathbb{E}_{H}\frac{
    \left\|
        \left(
            |\mathbf{x}\rangle\langle\mathbf{x}|
            \otimes \Pi_{\mathbf{x},H(\mathbf{x})}
        \right)
        \mathcal{A}^{H}
        |\phi_0\rangle
    \right\|^2
}{
    (2q+1)^{2n}
}.
\]
If we repeat the argument with our new projector $\Pi_{\mathbf{x},\mathbf{\Theta}}^{H(\mathbf{x} \ast \mathbf{\Theta})}$ the same change of variables occurs: the left hand side becomes
\[
\mathbb{E}_{H}\frac{
    \left\|
        \left(
            |\mathbf{x}\rangle\langle\mathbf{x}|
            \otimes \Pi_{\mathbf{x},H(\mathbf{x})}^H
        \right)
        \mathcal{A}^{H}
        |\phi_0\rangle
    \right\|^2
}{
    (2q+1)^{2n}
}.
\]
Note that here it is important that we give $\Pi_{\mathbf{x},\mathbf{\Theta}}$ access to the reprogrammed $H$, \emph{not} the original $H$: otherwise $\Pi_{\mathbf{x},\mathbf{\Theta}}$ and $\cA$ would be accessing different oracles. This makes intuitive sense since $\Pi_{\mathbf{x},\mathbf{\Theta}}$ is applied after the oracle is reprogrammed. With that the proof of the main measure-and-reprogram theorem goes through the same way.

\end{document}